\documentclass[aps,prd,nofootinbib,notitlepage,unsortedaddress]{revtex4-2}
\usepackage{enumerate}
\usepackage{dsfont}
\usepackage{amsfonts}
\usepackage{amssymb}
\usepackage[latin1]{inputenc}
\usepackage[T1]{fontenc}
\usepackage{graphicx}
\usepackage{mathtools}
\usepackage{bbm}
\usepackage{ulem}

\usepackage[dvipsnames]{xcolor}

\usepackage[colorlinks=true]{hyperref}
\usepackage{amsthm}

\newcommand{\R}{\mathbb R}

\newcommand{\D}{\text{d}}

\newcommand{\dd}{\partial}
\newcommand{\db}{\bar\partial}

 \newtheorem{prop}{Proposition}

\newtheorem{lem}[prop]{Lemma}
\newtheorem{theor}[prop]{Theorem}

\begin{document}
\title{Randers pp-waves}

\author{Sjors Heefer}
\email{s.j.heefer@tue.nl}
\affiliation{Department of Mathematics and Computer Science, Eindhoven University of Technology, Eindhoven, The Netherlands}
\author{Christian Pfeifer}
\email{christian.pfeifer@ut.ee}
\affiliation{Laboratory of Theoretical Physics, Institute of Physics, University of Tartu, Tartu, Estonia}
\affiliation{ZARM, University of Bremen, 28359 Bremen, Germany.}
\author{Andrea Fuster}
\email{a.fuster@tue.nl}
\affiliation{Department  of Mathematics and Computer Science, Eindhoven University of Technology, Eindhoven, The Netherlands}

\begin{abstract}
In this work we study Randers spacetimes of Berwald type and analyze Pfeifer and Wohlfarth's vacuum field equation of Finsler gravity for this class. We show that in this case the field equation is equivalent to the vanishing of the Finsler Ricci tensor, analogously to Einstein gravity. This implies that the considered vacuum field equation and Rutz's equation coincide in this scenario. We also construct all exact solutions of Berwald-Randers type to vacuum Finsler gravity, which turn out to be composed of a Ricci-flat, CCNV (covariantly constant null vector) Lorentzian spacetime, or Brinkmann space, and a 1-form defined by its covariantly constant null vector. Since the pp-waves are the most well-known metric in this class we refer to the found solutions as \textit{Randers pp-waves}.
\end{abstract}

\maketitle


\section{Introduction}
Finsler geometry is an extension of Riemannian geometry, in which the squared line element is not restricted to be quadratic in the displacements, and it is a natural framework allowing for a canonical definition of length for curves on a smooth manifold.\\

Historically, the possibility of considering this type of geometry was already discussed by Riemann himself in his famous habilitation lecture in 1854 \cite{Riemann1,Riemann2}. The first systematic study of such spaces, however, appeared only much later, in the dissertation thesis of Finsler in 1918 \cite{Finsler}. It is therefore hardly  surprising that general relativity was formulated on the basis of (pseudo-)Riemannian geometry, which was already well developed at the time. In fact, there is no fundamental physical reason to exclude the use of proper, non-Riemannian, Finslerian spacetimes to describe gravity  \cite{Tavakol_1986, Pfeifer_2019}. \\

The exploration of theories of gravity based on Finsler geometry only began much later, once it became clear that general relativity may not be the complete answer to our understanding of gravity. Today we know that general relativity is able to explain many observations with astonishing precision, but not at all scales. The theory faces substantial problems at the very large and very small scales. In the first case dark matter and dark energy have to be postulated, and the challenge in the latter scenario is to reconcile the theory with the principles of quantum mechanics. A more general theory of gravity, based on Finsler geometry, might shed light on some of these issues.\\

There are two particularly compelling examples of how Finsler geometry may lead to new insights on the nature of the gravitational interaction. The first one appears in the context of quantum gravity. Research suggests that the flat, classical limit of the quantum-mechanical description of gravity may be given by an effective semi-classical theory in which spacetime is a curved Finsler manifold. The geometrical structure of spacetime in this regime seems to be of a more general type and it is truly necessary to use Finsler geometry (or even Lagrange geometry) in order to describe it. This is intimately related to deformed Poincar\'e-symmetries and corresponding modified (energy-momentum) dispersion relations (MDRs). For instance, a certain Finsler geometry can be coupled to any MDR (satisfying some basic mathematical assumptions) in order to describe the corresponding dynamics \cite{Girelli:2006fw,Amelino-Camelia:2014rga,Letizia:2016lew}. This stands in contrast to Lorentzian geometry, which can only describe quadratic MDRs.

The second example comes from the coupling between gravity and fluids, respectively gases. The kinetic theory of gases describes the dynamics of multi-particle systems in terms of a scalar field on the tangent bundle, the so-called 1-particle distribution function (1PDF) \cite{Ehlers2011}. The coupling between gravity and such a gas usually involves averaging over the velocity distribution of its particles and leads to the Einstein-Vlasov equations \cite{Andreasson:2011ng}. Hence, although the velocity distribution of the gas is taken into account for its dynamics, it is averaged away when determining the gas  gravitational field. A direct coupling of the gas 1PDF to Finsler geometry circumvents this loss of information and gives rise to a gas gravitational field distribution described by Finsler geometry \cite{Hohmann:2019sni,Hohmann:2020yia}
\\

In this work we consider the (action-based) approach to Finsler gravity outlined in \cite{Pfeifer:2011xi,Hohmann_2019}. Structurally the theory is analogous to general relativity, in the sense that admissible, or "physical", spacetimes are those that satisfy a certain field equation. Solutions of Berwald type to Pfeifer and Wohlfarth's vacuum field equation have recently been found \cite{Fuster:2015tua,Fuster:2018djw,Caponio_2020}.  We analyze this equation for Randers spacetimes of Berwald type and prove that in this scenario it is formally identical to Einstein's vacuum field equations. As a result, the classification of its solutions reduces to the classification of vacuum solutions in general relativity admitting a covariantly constant $1$-form. The resulting Berwald-Randers spacetimes are composed of a Ricci-flat, CCNV (covariantly constant null vector) Lorentzian spacetime, or Brinkmann space \cite{Brinkmann1925}, and a 1-form defined by its covariantly constant null vector. Since the pp-waves are the most well-known metric in this class we refer to the found solutions as \textit{Randers pp-waves}. 

\section{Finsler geometry}\label{sec:FSFST}
Finsler geometry is a natural extension of Riemannian geometry \cite{Finsler,Bao,Szilasi}. Given the philosophy that the length of a curve is obtained by integrating the length of its tangent vector, Finsler geometry provides the most general way of assigning lengths to curves on a manifold. While in Riemannian geometry the length of a tangent vector is given by the metric-induced norm, in Finsler geometry this length is given by a so-called Minkowski norm; a much weaker and more general notion. \\

First of all some remarks about notation are in order. Throughout this work we will usually work in local coordinates, i.e., given a smooth manifold $M$ we assume that some chart $\phi:U\subset M\to \R^n$ is provided, and we identify any $p\in U$ with its image $\phi(p)\in\R^n$. For $p\in U$ each $Y\in T_pM$ (the tangent spaces to $M$) can be written as $Y = y^i\partial_i\big|_p$, where the tangent vectors $\partial_i \equiv \frac{\partial}{\partial x_i}$ furnish the chart-induced basis of $T_pM$. This provides natural local coordinates on the tangent bundle $TM$ via the chart
\begin{align}
\tilde\phi: \tilde U \to \R^n\times\R^n,\qquad \tilde U = \bigcup_{p\in U} \left\{p\right\}\times T_p M\subset TM,\qquad \tilde\phi(p,Y) = (\phi(p),y^1,\dots,y^n)\eqqcolon (x,y).
\end{align}
These local coordinates on $TM$ in turn provide a natural basis of its tangent spaces $T_{(x,y)}TM$, namely
\begin{align}
\bigg\{\frac{\partial}{\partial x^i} = \partial_i, \frac{\partial}{\partial y^i} = \bar{\partial}_i\bigg\}.
\end{align}

Next we will first introduce the basic notions of Finsler geometry for the positive definite case. The generalization to Lorentzian signature is non-trivial, however, and will be introduced afterwards.

\subsection{Finsler spaces of positive definite signature}\label{sec:FinslerSpaces}
A Finsler space is a pair $(M,F)$, where $M$ is a smooth manifold and $F$, the so-called Finsler function, is a map $F:TM\to[0,\infty)$ that satisfies the following axioms:
\begin{itemize}
	\item $F$ is (positively) homogeneous of degree one with respect to $y$:
	\begin{align}
	F(x,\lambda y) =\lambda F(x, y)\,,\quad \forall \lambda>0\,;
	\end{align}
	\item $F$ is strictly convex in $y$, i.e., the \textit{fundamental tensor}, with components $g_{ij} = \db_i\db_j \left(\frac{1}{2}F^2\right)$, is positive definite.
\end{itemize}
For each $x\in M$ the map $y\mapsto F(x,y)$ is what is known as a Minkowski norm\footnote{Not to be confused with the flat Lorentzian Minkowski metric.} on $T_xM$, i.e., a real-valued function that is positively homogeneous, strictly convex and smooth away from the zero vector. The homogeneity conditions ensure that the length of any curve $\gamma$, defined as
\begin{align}
L(\gamma)=\int  F(\dot{\gamma})\,\D \lambda = \int  F(x,\dot{x})\,\D \lambda,\qquad \dot{\gamma}=\frac{d\gamma}{d\lambda},
\end{align}
is invariant of the parameterization. A fundamental result that is essential for doing computations in Finsler geometry is Euler's theorem for homogeneous functions. It says that if $f:\R^n\to\R$ is (positively) homogeneous of degree $r$, i.e., $f(\lambda y) =\lambda^r f(y)$ for all $\lambda>0$, then $y^i\frac{\dd f}{\dd y_i}(y) = r f(y)$. In particular, this implies the identity
\begin{align}
g_{ij}(x,y)y^i y^j = F(x,y)^2.
\end{align}
Hence the length of curves is formally identical to the length in Riemannian geometry, the difference being that now the metric tensor may depend on the direction in addition to position.\\

The fundamental theorem of Riemannian geometry says that any Riemannian manifold admits a unique torsion free affine connection that is compatible with the metric, the Levi-Civita connection.  A similar statement is true in Finsler geometry, and this is sometimes called the fundamental lemma of Finsler geometry: it states that any Finsler space can be endowed with a canonical connection. An essential difference with Riemannian geometry is that the connection on a Finsler space is in general not a linear one. Let us therefore briefly recall the notion of a non-linear connection. A non-linear  (or Ehresmann) connection is a smooth decomposition of $TTM$ into a horizontal and a vertical subbundle,
\begin{align}
TTM = HTM \oplus VTM.
\end{align}
This provides the most general means of describing  parallel transport of vectors between tangent spaces, and, in particular, it allows one to define whether a curve $\gamma:I=(a,b)\to M$ is autoparallel (`straight'). Intuitively, we would like to call a curve straight whenever the velocity $\dot\gamma:I\to TM$ is `constant'. However, there is no unique way to say, \textit{a priori}, what `constant' means in this context, as each image point of $\dot\gamma$ lies in a different tangent space. As a matter of fact, as $\dot\gamma$, living in the tangent bundle, also contains all information about the base point $\gamma$, it could never be truly constant. Indeed, all we can ask is that  $\dot\gamma$ change only `parallel to $M$', and not in the direction of the fibres of $TM$. The rate of change of $\dot\gamma$, i.e.  $\ddot\gamma$, is an element of $TTM$. Therefore, in order to be able to say what we mean by a straight line we should split the directions in $TTM$ into a space $HTM$ of directions parallel to $M$ and a space of directions $VTM$ along the fibers of $TM$. We then say that a curve $\gamma:I\to M$ is autoparallel if $\ddot\gamma(\lambda)\in H_{\dot\gamma(\lambda)}TM$ for all $\lambda\in I$. The vertical subbundle $VTM$ is canonically defined on any smooth manifold, namely
\begin{align}
VTM = \text{span}\left\{\bar\partial_i\right\}.
\end{align}
However, there is in general not a preferred choice of the horizontal subbundle. In order to be able to speak about straight curves, in the most general sense, one thus needs to select one. In order to do so, a set of functions $N^i_j(x,y)$, the connection coefficients, may be specified, leading to the following horizontal subbundle of $TTM$.
\begin{align}
HTM = \text{span}\left\{\delta_i\equiv \partial_i - N^j_i\db_j\right\}.
\end{align}
Parallel transport of a vector field $V$ along $\gamma$ is then characterized by the parallel transport equation
\begin{align}
\label{eq:nonlinear.parallel.transport.eq}
\dot V^i + N^i_j(\gamma,V)\dot \gamma^j = 0\,,
\end{align}
and consequently, autoparallels are precisely the curves that satisfy
\begin{align}
\label{eq:nonlinear.geodesic.eq}
\ddot \gamma^i + N^i_j(\gamma,\dot \gamma)\dot \gamma^j = 0\,.
\end{align}
As mentioned, on generic smooth manifold there is no canonical choice of the connection\footnote{From now on we will refer to the connection coefficients $N^i_j$ simply as the connection.} $N^i_j$, but any Finsler metric induces one, the \textit{Cartan non-linear connection}. This is the unique homogeneous (non-linear) connection on $TM$ that is (smooth on $TM\setminus\{0\}$,) torsion-free and compatible with $F$. This \textit{Cartan non-linear connection} is given in terms of the Finsler function $F$ by
\begin{align}
N^i_j(x,y) = \frac{1}{4}\bar{\partial}_j \bigg(g^{ik}\big(y^l\partial_l\bar{\partial}_k F^2 - \partial_k F^2\big)\bigg)\,
\end{align}
and may be viewed as a generalization of the Levi-Civita connection to Finsler spaces. The autoparallel curves of the non-linear connection coincide with the geodesics (locally length-minimizing curves) on $M$. The curvature tensor, curvature scalar and the Finsler Ricci tensor 
of $(M,F)$ are defined as
\begin{align}\label{eq:definition_curvatures}
R^i{}_{jk}(x,y) = -[\delta_j,\delta_k]^i =  \delta_j N^i_k(x,y)-\delta_k N^i_j(x,y),\qquad \text{Ric}(x,y) = R^i{}_{ij}(x,y)y^j,\qquad R_{ij}(x,y) = \frac{1}{2}\db_i \db_j\text{Ric}.
\end{align}

\subsection{Berwald spaces and the Riemannian limit}\label{sec:Berwald}

A Berwald space is a Finsler space $(M,F)$ for which the Cartan non-linear connection is in fact a linear connection on $TM$.\footnote{See \cite{Szilasi2011} for an overview of the various equivalent characterizations of Berwald spaces and \cite{Pfeifer:2019tyy} for an explicit procedure to construct Berwald spaces and spacetimes. The latter was used in \cite{Hohmann:2020mgs} to find all homogeneous and isotropic Berwald Finsler geometries.} What this means is that the connection coefficients are of the form
\begin{align}
N^i_j(x,y) = \Gamma^i_{jk}(x)y^k
\end{align}
for a set of functions $\Gamma^i_{jk}:M\to\R$. From the transformation behavior of $N^i_j$ it can be inferred that the functions $\Gamma^i_{jk}$ have the correct transformation behavior to be the Christoffel symbols of a (torsion-free) affine connection on $M$. We will refer to this affine connection as the associated affine connection, or simply \textit{the} affine connection on the Berwald space. 
The parallel transport \eqref{eq:nonlinear.parallel.transport.eq} and autoparallel equations \eqref{eq:nonlinear.geodesic.eq} reduce in this case to the familiar equations
\begin{align}
\dot V^i + \Gamma^i_{jk}(\gamma)\dot \gamma^j V^k = 0, \qquad \ddot \gamma^i + \Gamma^i_{jk}(\gamma)\dot \gamma^j \dot \gamma^k = 0
\end{align}
in terms of the Christoffel symbols. A straightforward calculation reveals that the curvature tensors of a Berwald space can be written as follows
\begin{align}
\label{eq:symm_ricci}
R^j{}_{kl} = \bar R_i{}^j{}_{kl}(x)y^i, \qquad \text{Ric} = \bar R_{ij}(x)y^i y^j, \qquad R_{ij} = \frac{1}{2}\left(\bar R_{ij}(x) + \bar R_{ji}(x)\right),
\end{align}
in terms of the curvature tensor of the associated affine connection
\begin{align}\label{eq:Curv}
	\bar R_l{}^i{}_{jk}= 2\partial_{[j} \Gamma^i_{k]l} + 2\Gamma^i_{m[j}\Gamma^m_{k]l}
\end{align} and its Ricci tensor $\bar R_{lk} = \bar R_l{}^i{}_{ik}$\footnote{We use the notations $T_{[ij]} = \frac{1}{2}\left(T_{ij}-T_{ji}\right)$ and $T_{(ij)} = \frac{1}{2}\left(T_{ij}+T_{ji}\right)$ for (anti-)symmetrization.}, defined in the usual way. In fact, for \textit{positive definite} Berwald spaces, one even has $R_{ij} = \frac{1}{2}\left(\bar R_{ij} + \bar R_{ji}\right) = \bar R_{ij} $, but this does \textit{not} extend to Finsler spacetimes, as will be  discussed in some more detail in the next section. \\

Finsler geometry reduces to Riemannian geometry when the fundamental tensor $g_{ij}(x,y)=g_{ij}(x)$ is independent of direction $y$, i.e., if the fundamental tensor is a Riemannian metric. Equivalently, the space is Riemannian if $F^2$ is quadratic in the $y$-coordinates. In this case the non-linear connection is actually linear, so that, in particular, any Riemannian manifold is Berwald. In fact, the associated linear connection is in this case nothing more than the Levi-Civita connection of the Riemannian metric.

\subsection{Finsler spacetimes}


The generalization of positive definite Finsler geometry to indefinite, for instance Lorentzian, signature is not completely trivial. To see the basic issue, note that if the fundamental tensor $g_{\mu\nu}$ has Lorentzian signature then there will be (non-zero) null vectors $v\in T_x M$ for which $g_{\mu\nu}v^\mu v^\nu=0$. Then $F(x,v)=\sqrt{g_{\mu\nu}v^\mu v^\nu}=\sqrt{0}$, even though $v\neq 0$, so $F$ can never be smooth everywhere on $TM\setminus 0$, which was one of the axioms of a Finsler space. Moreover, for spacelike (or timelike, depending on the convention) directions $w$, $F(x,w)$ will even be imaginary. Thus some things clearly need to be modified in order to give an acceptable definition of a Finsler spacetime. Multiple approaches are possible. One classical approach \cite{Beem} is to work with $L = F^2$ instead of $F$. Another is to restrict the domain of definition of $F$, for instance to those $(x,y)$ for which $F(x,y)^2 = g_{\mu\nu}(x,y)y^\mu y^\nu>0$ \cite{Asanov}. Also combinations of the two approaches and even additional variations have been proposed \cite{Pfeifer:2011tk, Pfeifer:2011xi, Javaloyes2014-1, Javaloyes2014-2}. What we will do in this work is simply replace the subbundle $TM\setminus\{0\}\subset TM$ by a generic conic\footnote{The property of being conic means that if $(x,y)\in\mathcal A$ then also $(x,\lambda y)\in\mathcal A$, for any $\lambda>0$.} subbundle $\mathcal A\subset TM$.  One might say this is the \textit{weakest} definition of a Finsler spacetime, in the sense that it allows for the most examples. This is of course the most suitable for our present purpose, which is to find the most general solutions to the \mbox{vacuum} field equations of Finsler gravity in the Berwald-Randers class. In this way we guarantee that our results do not depend on a particular definition, and that we indeed find the most general solutions. Given any other specific definition of a Finsler spacetime, the most general solutions will then still be a subset of the ones we present later.\\

Thus we will be using the following definition. A Finsler spacetime is a triple $(M,\mathcal A,F)$, where $M$ is a smooth manifold, $\mathcal A$ is a conic subbundle of $TM$ (with `non-empty' fibers) and $F$, the so-called Finsler function, is a map $F:\mathcal A\subset TM\to [0,\infty)$ that satisfies the following axioms:
\begin{itemize}
	\item $F$ is (positively) homogeneous of degree one with respect to $y$:
	\begin{align}
	F(x,\lambda y) =\lambda F(x, y)\,,\quad \forall \lambda>0\,;
	\end{align}
	\item The \textit{fundamental tensor}, with components $g_{\mu\nu} = \db_\mu\db_\nu \left(\frac{1}{2}F^2\right)$, has Lorentzian signature on $\mathcal{A}$.
\end{itemize}

The discussion and results (for the connection, curvature tensors, etc.) treated in sections \ref{sec:FinslerSpaces} and \ref{sec:Berwald} apply verbatim for Finsler spacetimes, asumming that we only consider points $(x,y)\in\mathcal A$. In particular, since the last equality in Eq.\,\eqref{eq:symm_ricci} will prove to be of particular interest to us, we restate it as a lemma.
\begin{lem}
\label{lem:RicciTensors}
The Finsler Ricci tensor of a Berwald space(time) coincides with the symmetrization of the conventional Ricci tensor of the associated affine connection.
\end{lem}

We remark that in the case that $\mathcal A = TM\setminus\{0\}$, symmetrization is in fact not necessary, as the Ricci tensor is then automatically symmetric. In the positive definite case this follows immediately from Szabo's theorem, which states that the affine connection of any Berwald space is Riemann-metrizable \cite{Szabo}. In the case of Finsler spacetimes the situation is more complicated.
In general Szabo's theorem does {\textit{not}} hold in this setting, as has been demonstrated in~\cite{Fuster_2020}. Therefore, there is no reason to expect the Ricci tensor $\bar R_{ij}$ to be symmetric. In order to shed a bit more light on this, we may use the first Bianchi identity, $\bar R_{[l}{}^i{}_{jk]}=0$ (which holds as the connection is torsion-free) to deduce that 
\begin{align}
	2 \bar R_{[lk]} = \bar R_{i}{}^i{}_{lk}\,.
\end{align} 
When the connection is the Levi-Civita connection of a pseudo-Riemannian metric, the symmetries of the Riemann tensor dictate that the right hand side vanish identically and hence in that scenario the Ricci tensor is always symmetric. Here, however, the curvature tensor, as defined in \eqref{eq:Curv}, does not have the same symmetries as in Riemannian geometry, because the connection, in general, does not come from a pseudo-Riemannian metric. Indeed in the cited article we present explicit examples where these symmetries do not hold and the Ricci tensor is indeed not symmetric.

One can prove that if $\mathcal A = TM\setminus\{0\}$ then the Ricci tensor must be symmetric after all, as illustrated in all detail in the same article, but in general it is not and hence symmetrization in the last equallity of \eqref{eq:symm_ricci} is really necessary.

\subsection{Randers spacetimes}

A Randers metric \cite{Randers} is a Finsler function of the form $F =  \alpha + \beta$, where the variables $\alpha(x,y) = \sqrt{a_{\mu\nu}(x)y^\mu y^\nu}$ and $\beta(x,y) = b_\mu(x) y^\mu$ are defined in terms of a (pseudo-)Riemannian metric\footnote{By abuse of language $\alpha$ is often referred to as a (pseudo-)Riemannian metric.} $a = a_{\mu\nu}\D x^\mu \D x^\nu$ and a 1-form $b = b_\mu \D x^\mu$. By convention, all indices are raised and lowered with the $a_{\mu\nu}$ and We denote the norm of the 1-form by $|b|= \sqrt{g^{\mu\nu}b_\mu b_\nu}$. Randers metrics belong to the class of $(\alpha,\beta)$-metrics, which are Finsler functions of the form $F = \alpha\, \phi(\beta/\alpha)$ with $\alpha$ and $\beta$ as defined above and $\phi$ a scalar function. For Randers metrics, the function $\phi$ is given by $\phi(s) = 1+s$. Randers metrics have been studied extensively in both the  positive definite as well as the Lorentzian context, and have a wealth of applications in physics. \\

A well-known result states that if $\alpha$ is positive definite then $F=\alpha+\beta$ is a Finsler metric if and only if $|b|<1$. A crucial ingredient that leads to this conclusion is the Cauchy-Schwarz inequality $|\beta|\leq |b| \alpha $, which does not extend in this simple form to indefinite signatures. Therefore in the latter case the the situation becomes somewhat different. It can be shown via the matrix determinant lemma that for $\alpha$-timelike (i.e. $a_{\mu\nu}y^\mu y^\nu>0$) vectors the determinant of the fundamental tensor of a Randers metric is given by
\begin{align}\label{eq:determinant_Randers}
\det g =  \left(\frac{\alpha+\beta}{\alpha}\right)^{n+1}\det a,
\end{align}
where $n=\dim M$. This expression has in general no well-defined limit as $\alpha\to 0$, so from this it is clear that, in the Lorentzian case, the subbundle $\mathcal A$ can only include points $(x,y)$ for which $\alpha\neq 0$. In order to have a connected subbundle, one might propose to work on the subbundle consisting of $\alpha$-timelike vectors, i.e., the timecone of $\alpha$. However, this does not work in general, and one has to restrict $\mathcal A$ just a little bit further.

\begin{prop}\label{prop:Randers_well_defined_on_future_timecone}
Given a Lorentzian metric $\alpha = \sqrt{a_{\mu\nu}y^\mu y^\nu}$ on a manifold $M$ with time-orientation $T$ and a future-pointing 1-form $\beta$ that is either null or timelike with respect to $a$, the Randers metric $F=\alpha+\beta$ defines a Finsler spacetime on 
\begin{align}
\mathcal A = \left\{(x,y)\in TM\,: a_{\mu\nu}(x)y^\mu y^\nu>0,\,\, a_{\mu\nu}(x)T^\mu(x)y^\nu>0\right\},
\end{align}
i.e. the forward timecone of a.
\end{prop}

The proof, presented in Appendix \ref{app:proof}, relies on showing that $\beta= b_i y^i >0$ holds everywhere on $\mathcal A$. \\

The following is a well-established result in the positive definite case, which holds in the indefinite case as well; see e.g. \cite{Pfeifer:2019tyy}, where conditions for a Finsler spacetime to be of Berwald type are presented.

\begin{prop}
\label{prop:coinciding_spray}
A Randers space, $F=\alpha+\beta$ is Berwald if and only if the 1-form $\beta$ is covariantly constant w.r.t. the Levi-Civita connection of $a$. If so, the affine connection of $F$ is identical to the Levi-Civita connection of $a$.
\end{prop}

\section{Vacuum Field Equation}\label{sec:FinslerGravity}

In the context of Finsler gravity, two types of field equations have been proposed in the literature. First, tensorial field equations considering the Finsler metric as fundamental variable \cite{Chang:2009pa,Asanov,Kouretsis:2008ha,Stavrinos2014,Voicu:2009wi,Minguzzi:2014fxa}, and second, scalar field equations for which the Finsler Lagrangian is the fundamental variable \cite{Rutz,Pfeifer:2011xi,Chen-Shen}. We adopt the latter view, and in addition demand the field equation to be variational\footnote{This constraints the field equation to be a scalar equation on the tangent bundle.}, meaning that it can be derived from an action. The first Finsler generalization of Einstein's equations in a scalar fashion was suggested by Rutz \cite{Rutz}. However, Rutz's equation is \textit{not} variational; its variational completion \cite{Hohmann_2019} \cite{Akbar-Zadeh1988} turns out to be the Finsler gravity equation proposed in \cite{Pfeifer:2011xi} (for positive definite Finsler spaces the analogous equation was presented in \cite{Chen-Shen}), which we therefore consider as our departure point. In the case of Berwald spacetimes it reduces to \cite{Fuster:2018djw}
\begin{align}\label{eq:BEFEs}
	\left(  g^{\mu\nu} - \frac{3}{F^2} y^\mu y^\nu \right) R_{\mu\nu}  = 0\,,
\end{align}
where $R_{\mu\nu}$ is the Finsler Ricci tensor and since we are in a Berwald setting we have $R_{\mu\nu} = \bar R_{(\mu\nu)}(x)$, see Eq.\,(\ref{eq:symm_ricci}). Note that the vanishing of the Finsler Ricci tensor is a sufficient condition for a Berwald spacetime to be a solution to Eq.\,\eqref{eq:BEFEs}.  In fact, for Randers spaces we have a stronger result. 

\begin{prop}\label{prop:Randers_EFEs}
For Berwald-Randers spacetimes, the field equation, Eq.\,\eqref{eq:BEFEs} is equivalent to the vanishing of the Finsler Ricci tensor $R_{\mu\nu}(x) = \bar R_{\mu\nu}(x) =\bar R_{(\mu\nu)}(x)=0$.
\end{prop}

\begin{proof}
To see this we first compute the fundamental tensor for a Randers spacetime $F = \alpha+\beta$. It reads
\begin{align}
    g_{\mu\nu} = \frac{F}{\alpha}a_{\mu\nu} + b_\mu b_\nu + (b_\mu y_\nu + b_\nu y_\mu) \frac{1}{\alpha} -  y_\mu y_\nu \frac{\beta}{\alpha^3}.
\end{align}
It can be checked that its inverse is given by
\begin{align}
    g^{\mu\nu} = a^{\mu\nu}\frac{\alpha}{F} - (b^\nu y^\mu + b^\mu y^\nu)\frac{\alpha}{F^2}+ y^\mu y^\nu \frac{(\beta  + |b| \alpha)}{F^3}\,.
\end{align}
Substituting the expression for the $g^{\mu\nu}$ into the Berwald field equation, \eqref{eq:BEFEs}, we find that for Berwald-Randers spacetimes the field equation become equivalent to
\begin{align}
     R \alpha (\alpha + \beta)^2  - 2\alpha (\alpha + \beta) R_{\mu\nu} b^\mu y^\nu  + (\beta + \alpha |b|^2){R}_{\mu\nu}y^\mu y^\nu - 3 ( \alpha + \beta){R}_{\mu\nu} y^\mu y^\nu &= 0,
\end{align}
which is a a polynomial equation in $\alpha$. We can separate this equation into the rational part (even powers in $\alpha$) and irrational part (odd powers in $\alpha$) as
\begin{align}\label{eq:alphapowers}
   R \alpha^3 + (R \beta^2 - 2 \beta R_{\mu\nu}y^\mu b^\nu + (|b|^2-3) R_{\mu\nu} y^\mu y^\nu )\alpha =  2 ( R_{\mu\nu} b^\mu y^\nu - R \beta )\alpha^2 + 2 \beta{R}_{\mu\nu} y^\mu y^\nu \,.
\end{align}
In order for this equation to hold, the left-hand side must become polynomial in $y$. However, since $\alpha=\sqrt{a_{\mu\nu}y^\mu y^\nu}$ is built from the components of a non-degenerate metric $a_{\mu\nu}$ this is only possible if both the rational and irrational part vanish separately.

To verify this last statement before we continue, we rewrite the equation above as
\begin{align}\label{eq:irrat1}
	P \alpha = Q\,,
\end{align} 
where $Q = Q_{\mu\nu\rho}(x)y^\mu y^\nu y^\rho$ and $P = P_{\mu\nu}(x)y^\mu y^\nu$ are third and second order polynomials in $y$.
Thus, assuming $P$ is non-vanishing then $\alpha = \frac{Q}{P}$ everywhere on $\mathcal{A}$. In other words $\alpha$ is a \textit{rational} function everywhere on $\mathcal{A}$. This is however not possible since $\alpha$ is the square root of a pseudo-Riemannian metric, which is in particular non-degenerate and by definition non-vanishing on $\mathcal{A}$. Hence we must have $P=0$, by contradiction.

To continue our proof we now focus on the vanishing of the rational part $Q$, which implies
\begin{align}\label{eq:feqrat}
    2 ( R_{\mu\nu} b^\mu y^\nu - R \beta )\alpha^2 + 2 \beta{R}_{\mu\nu} y^\mu y^\nu = 0\,.
\end{align}
Taking two derivatives with respect to $y$ and taking the trace with $a$ yields
\begin{align}
   4(5 R \beta - 8 R_{\mu\nu}b^\mu y^\nu) = 0.
\end{align}
Substituting this into equation \eqref{eq:feqrat} yields
\begin{align}
    \label{eq:last_eq_proof}
    \beta\left(\frac{3}{4}R \alpha^2  - 2 {R}_{\mu\nu} y^\mu y^\nu\right)=0.
\end{align}
After applying other two $y$-derivatives to the term in parenthesis and contracting again with $a$ we get $2R=0$. Substituting this back into \eqref{eq:last_eq_proof} yields $R_{\mu\nu}y^\mu y^\nu =0$ and hence $R_{\mu\nu}=0$.
\end{proof}
Notice that in this specific scenario, Berwald-Randers, the vacuum field equations of Finsler gravity are formally identical to the vacuum field equations of general relativity. Moreover, they coincide with Rutz's proposal for vacuum field equation in Finsler gravity \cite{Rutz}, $\text{Ric}=0$, which amounts to the vanishing of the Finsler Ricci tensor, $R_{\mu\nu}=0$, see Eq.\,\eqref{eq:definition_curvatures}. The analogous result for positive definite Finsler spaces has been proven in \cite{Chen-Shen}. 

Next we will investigate solutions of the field equation. The previous results show that in Berwald-Randers spacetimes, the metric $a$ is determined by the Einstein vacuum equations $\bar R_{\mu\nu} = 0$ and $1$-form $b=b_\mu dx^\mu$ by the condition that it is covariantly constant w.r.t.\,metric $a$, $\nabla_\mu b_\nu = 0$.

\section{Exact Berwald-Randers Solutions}

We now have all the background required to prove our main results. The following theorem classifies all Berwald-Randers solutions to the vacuum FEFE (Finsler generalization of Einstein's field equation), Eq.\,\eqref{eq:BEFEs}, that are  Berwald in terms of solutions to the classical EFEs.

\begin{theor}
\label{theor:main}
Let $(M,\mathcal A, F)$ be a Randers spacetime of Berwald type, $F = \alpha + \beta$. Then $F$ is an exact solution to the vacuum FEFE if and only if $\alpha$ is a solution to the vacuum EFEs.
\end{theor}
\begin{proof}
Since $F$ is Berwald-Randers, Proposition \ref{prop:coinciding_spray}  implies that its affine connection coincides with the Levi-Civita connection of the Lorentzian metric $\alpha$. With Lemma \ref{lem:RicciTensors} it thus follows that the Finsler Ricci tensor coincides with the Ricci tensor of $\alpha$, and since for Lorentzian resp.\,Randers spaces the property of being a solution is equivalent to the vanishing of the relevant Ricci tensor, by Proposition \ref{prop:Randers_EFEs}, $\alpha$ solves the vacuum EFEs if and only $F$ solves the vacuum FEFE.
\end{proof}

Taking into account Prop. \ref{prop:coinciding_spray}, we may also formulate the theorem in the following way.

\begin{theor}
\label{theor:main2}
Let $(M,\mathcal A, F)$ be a Finsler spacetime of Randers type, $F = \alpha + \beta$. Then the following are equivalent:
\begin{enumerate}
\item $F$ is an exact Berwald solution to the vacuum FEFE
\item $\alpha$ is an exact solution to the vacuum EFEs and $F$ is Berwald
\item $\alpha$ is an exact solution to the vacuum EFEs and $\beta$ is parallel w.r.t. $\alpha$.
\end{enumerate}
\end{theor}

An immediate consequence of the theorem is that a Randers metric of Berwald type can only be a solution to the Finsler gravity field equation provided $\alpha$ admits a parallel 1-form. This strongly restricts the set of possibilities for $\alpha$, which we investigate in detail in the next section.
 
\section{Randers pp-waves}
\label{sec:Randers_pp_waves}
It has been shown \cite{Batista_2014} that if a 4-dimensional Ricci-flat Lorentzian manifold $(M,g)$ admits a covariantly constant 1-form $\omega$, then either $g$ is the flat Minkowski metric (and hence the components of $\omega$, in suitable coordinates, must be constant), or $\omega$ is null and the metric can be written, in local coordinates $(u,v,x^1,x^2)$, with $u=(1/\sqrt{2})(x^0-x^3)$ and $v=(1/\sqrt{2})(x^0+x^3)$ light-cone coordinates, as
\begin{align} \label{eq:CCNV}
ds^2=-2\,\D u\left(\D v + H(u,x)\, \D u + \, W_a(u,x)\D x^a\right)  +h_{ab}(u,x) \D x^a \D x^b,\qquad \omega = \D u,\qquad a,b=1,2,
\end{align}
where $H$, $W_a$ are real metric functions and $h_{ab}$ is a $2$-dimensional Riemannian metric. We focus our attention on the non-trivial case given by Eq.\,\eqref{eq:CCNV}. Lorentzian spacetimes of this form are called Brinkmann spaces\footnote{Note that these are not necessarily Ricci flat.}, and are also referred to as CCNV (covariantly constant null vector) spacetimes; the standard pp-waves are obtained when the transverse metric is Euclidean. Brinkmann spaces in higher dimensions are also of this form, see for example \cite{Leistner2016,blanco2011,Mcnutt2009}.\\

We choose $\alpha = \D s^2(y,y)$ as above and $\beta = \D u(y) = y^u$, and consider the following Randers metric,
\begin{align}\label{eq:CCNV_Finsler}
F = \alpha + \beta = \sqrt{-2y^u \left(y^v + H(u,x)\, y^u + \,W_a(u,x)\,y^a\right) +h_{ab}(u,x) y^a y^b}+ y^u. 
\end{align}

By Proposition \ref{prop:Randers_well_defined_on_future_timecone},  any such $F$ defines a Finsler spacetime on the forward timecone of $\alpha$ if and only if $\D u$ is future-oriented, which we can always achieve (at least locally) by choosing an appropriate time-orientation. We can use these observations to strengthen our theorem.

\begin{theor}
Let $(M,F)$ be a Randers spacetime, $F = \alpha + \beta$. Then $F$ is an exact Berwald solution to the FEFE if and only if one of the following statements is true:
\begin{itemize}
\item There exist local coordinates such that
\begin{align}
F = \alpha + \beta = \sqrt{-2y^u \left(y^v + H(u,x)\, y^u + \,W_a(u,x)\,y^a\right) +h_{ab}(u,x) y^a y^b}+ y^u,
\end{align}
where $\alpha$ is a vacuum solution to Einstein gravity.
\item There exist local coordinates such that
\begin{align}
F = \alpha + \beta = \sqrt{-(y^0)^2+ (y^1)^2+ (y^2)^2+ (y^3)^2}+ b_\mu y^\mu, \qquad b_\mu = \text{const.}
\end{align}
\end{itemize}
\end{theor}
In particular, the second item in the theorem may be viewed as a Randers analogue of the Bogoslovsky/Very Special Relativity (VSR) line element \cite{Cohen:2006ky,Gibbons:2007iu,Bogoslovsky1973,Bogoslovsky}. The first item, may be viewed as a Randers analogue of \textit{Very General Relativity (VGR)} spacetimes \cite{Fuster:2018djw} and when $W_a=0$ and $h_{ab}=\delta_{ab}$, it reduces to a Randers analogue of the Finsler pp-waves introduced in \cite{Fuster:2015tua}. In practice (to the best of our knowledge) the most general solution to Einstein's vacuum field equations that can appear here, i.e.\,of the form \eqref{eq:CCNV}, is the gyratonic pp-wave \cite{gyraton, Maluf2018}, which belongs to the VSI class of spacetimes \cite{Pravda:2002us}. In this case the transverse metric is Euclidean,  $h_{ab} = \delta_{ab}$, and the vacuum Einstein equations read
\begin{align}
    \nabla^2 H &= -\frac{1}{2}\left(\partial_{x^2}W_1 - \partial_{x^1}W_2\right)^2 + \, \partial_u\left(\partial_{x^1}W_1 + \partial_{x^2}W_2\right)\\
   0
  &= \partial_{x^2}\left(\partial_{x^2}W_1 - \partial_{x^1}W_2\right)  \\
    0&=\partial_{x^1} \left(\partial_{x^2}W_1 - \partial_{x^1}W_2
    \right),
\end{align}
where $\nabla^2 = \partial_{x^1}^2 + \partial_{x^2}^2$. The latter two equations restrict the form of metric functions $W_1, W_2$, and the first one determines $H$ for given $W_1, W_2$. 

\section{Discussion}

In this work we study Randers spacetimes of Berwald type and analyze Pfeifer and Wohlfarth's vacuum field equation of Finsler gravity for this class. We show that in this case the field equation is equivalent to the vanishing of the Finsler Ricci tensor, analogously to Einstein gravity. This implies that the considered vacuum field equation and Rutz's equation coincide in this scenario. Then we construct all exact solutions of Berwald-Randers type to vacuum Finsler gravity, which turn out to be composed of a Ricci-flat, CCNV (covariantly constant null vector) Lorentzian spacetime, or Brinkmann space, and a 1-form defined by its covariantly constant null vector. Since the pp-waves are the most well-known metric in this class we refer to the found solutions as \textit{Randers pp-waves}. \\

Interestingly, such a Randers spacetime has the same affine connection, and hence the same geodesics, as the \mbox{corresponding} CCNV Lorentzian spacetime (see Proposition \ref{prop:coinciding_spray}). 
Not only that, this also holds for the CCNV subclass of Finsler VSI spacetimes presented in \cite{Fuster:2018djw}, which include the Finsler pp-waves \cite{Fuster:2015tua}. The latter case has also been discussed recently in \cite{Elbistan:2020mca}. 
In Finsler geometry it is not uncommon to encounter different Finsler metrics that nevertheless share the same geodesics. 
This remarkable feature of Finsler geometry was first pointed out by Tavakol and Van den Bergh \cite{Tavakol_1986}. However, having the same affine structure does not imply that the causal structure coincides as well. For example, the set of causal (timelike and null) directions of the Randers pp-waves does not coincide with that of the classical pp-waves, due to their different Finsler functions. \\

The natural question that arises is whether and how such spacetimes can be physically distinguished from one another. In particular, is there any kind of physical observable that could distinguish Randers pp-waves from Finsler pp-waves and/or Lorentzian pp-waves? Such questions cannot be answered without a consistent observer framework, see for example \cite{Gurlebeck:2018nme,Raetzel:2010je}. These studies suggest that time dilations between observers moving relatively to each other would yield different results for Randers and Lorentzian pp-wave spacetimes, due to the different normalization of timelike geodesics. Similarly, speed light measurements may not coincide for different observers  due to the modified null condition for light rays \cite{Bernal_2020}, although this strongly depends on the observer model employed \cite{Gurlebeck:2018nme}. These important questions will be addressed in future work. \\

A somewhat related issue pertains to geodesic completeness. It is known that any Lorentzian compact pp-wave is geodesically complete, and this reduces to plane waves in the Ricci flat case \cite{Leistner2016}. Whether these features also hold for Finslerian versions of pp-waves such as the ones in this work and in \cite{Fuster:2015tua} remains to be studied. However, the fact that all of these spacetimes share the same geodesics could point in that direction. \\

Both our Randers and the previously known $m$-Kropina solutions in \cite{Fuster:2018djw} belong to the class of $(\alpha,\beta)$-metrics. A natural next step would be to study the existence of generic solutions of $(\alpha,\beta)$-type, i.e.\,with a defining Finsler function $F = \alpha\, \phi(\beta/\alpha)$ with arbitrary scalar $\phi$. Even beyond this class, it would also be of interest to explore solutions in the form of Finsler $b$-spaces, which describe Lorentz violating particle kinematics derived from the standard model extension \cite{Kostelecky:2011qz} and certain classical mechanical systems \cite{FOSTER2015164}. Such metrics are characterized by a Finsler function $F = \alpha + \tilde\beta$, where $\tilde\beta = \sqrt{\beta^2 - |b|^2\alpha^2}$ and $b$ is an arbitrary 1-form. Note that Randers and $b$-spaces coincide for a null 1-form $b$.\\


Finally, it would be interesting to investigate Berwald-Randers spacetimes in higher dimensions. In this case a Brinkmann metric also induces a Berwald-Randers spacetime, in a straightforward generalization of expression \eqref{eq:CCNV_Finsler}. The vacuum field equations in higher-dimensional Finsler gravity, although not formally studied yet, are very likely to resemble the 4D case (up to numerical factors related to the considered number of dimensions). In such a scenario, spacetimes generalizing \eqref{eq:CCNV_Finsler} with a Ricci-flat Brinkmann metric would immediately be solutions of higher-dimensional Finsler gravity.


\begin{acknowledgments}
C.\,Pfeifer was supported and funded by the Estonian Ministry for Education and Science through the Personal Research Funding Grant PSG489, as well as the European Regional Development Fund through the Center of Excellence TK133 ``The Dark Side of the Universe'' and the Deutsche Forschungsgemeinschaft (DFG, German Research Foundation) - Project Number 420243324. A.\,Fuster thanks S.\,Hervik for his feedback on CCNV spacetimes and A.\,Ach\'{u}carro for everything. A.\,Fuster especially thanks D.\,Bao, whose question about the existence of Randers gravitational waves inspired this article. The work of A.\,Fuster is part of the research program of the Foundation for Fundamental Research on Matter (FOM), which is financially supported by the Netherlands \mbox{Organisation} for Scientific Research (NWO). The authors would like to acknowledge networking support by the COST Action CA18108, supported by COST (European Cooperation in Science and Technology).
\end{acknowledgments}
\appendix

\section{Proof of Proposition \ref{prop:Randers_well_defined_on_future_timecone}}\label{app:proof}

Here we include the proof that the Randers metric $F=\alpha+\beta$, where $\alpha$ is Lorentzian and $\beta$ is future-pointing and not spacelike, defines a Finsler spacetime on the forward timecone of $\alpha$.\\

The lemma below shows that $F = \alpha + \beta$ is strictly positive on the forward timecone of $\alpha$, so that $\det g$ is strictly positive as well, by Eq.\,\eqref{eq:determinant_Randers}. The remainder of the proof is the same as in the positive definite scenario. Introduce a parameter $\lambda$ and consider the family of functions $F_\lambda = \alpha+\lambda\beta$ for $\lambda\in [0,1]$. For each value of $\lambda$, $\det g_\lambda$ is strictly positive and so $g_\lambda$ has constant signature for $\lambda\in [0,1]$. In particular $g_{\lambda=1}$ has the same signature as $g_{\lambda=0}$, which is just the statement that $g$ is Lorentzian. Since the homogeneity of $F$ is clear, this completes the proof that $F$ defines a Finsler spacetime on the forward timecone of $\alpha$.

\begin{lem}\label{lem:Randers_well_defined_on_future_timecone}
Let $y$ be timelike, $b$ timelike or null, and $y$ and $b$ both future-oriented, all with respect to some time-oriented Lorentzian matrix $g_{\mu\nu}$ with index convention $(+,-,\dots,-)$. Then $g_{\mu\nu}y^\mu b^\nu>0$.
\end{lem}
\begin{proof}
Let the time orientation be defined in terms of a nowhere vanishing timelike vector field $v$. Then future-orientation of $y$ and $b$ means that we have $g_{\mu\nu}y^\mu v^\nu>0$ and $g_{\mu\nu}b^\mu v^\nu\geq 0$. Now notice that we can always pick a basis of $T_xM$ such that $v$ has coordinates $(v^0,0,\dots,0)$ and $g_{\mu\nu}=\eta_{\mu\nu}$ is the Minkowski metric (first pick a basis such that the metric is Minkowski, then apply an appropriate Lorentz boost to make all components of $v$ other then the first vanish, and then possibly apply a reflection in the $x^0$ coordinate). In this basis we write $\vec y = (y^1, y^2, ...y^n)$. With this notation the future-pointing properties of $y$ and $b$ translate to $y^0>0$ and $b^0>0$ and the remaining properties can be stated as $|y^0|>|\vec y|$ and $ |b^0|\geq |\vec b|$. Combined they read $y^0>|\vec y|$ and $ b^0\geq |\vec b|$. Then we have
\begin{align}
\eta_{\mu\nu}y^\mu b^\nu &= y^0b^0 - \vec y\cdot\vec b \geq y^0b^0 - |\vec y||\vec b| > y^0b^0 - y^0 b^0 = 0.
\end{align}
\end{proof}


\bibliography{Randers_pp_waves_Rev1}

\begin{thebibliography}{56}%
\makeatletter
\providecommand \@ifxundefined [1]{%
 \@ifx{#1\undefined}
}%
\providecommand \@ifnum [1]{%
 \ifnum #1\expandafter \@firstoftwo
 \else \expandafter \@secondoftwo
 \fi
}%
\providecommand \@ifx [1]{%
 \ifx #1\expandafter \@firstoftwo
 \else \expandafter \@secondoftwo
 \fi
}%
\providecommand \natexlab [1]{#1}%
\providecommand \enquote  [1]{``#1''}%
\providecommand \bibnamefont  [1]{#1}%
\providecommand \bibfnamefont [1]{#1}%
\providecommand \citenamefont [1]{#1}%
\providecommand \href@noop [0]{\@secondoftwo}%
\providecommand \href [0]{\begingroup \@sanitize@url \@href}%
\providecommand \@href[1]{\@@startlink{#1}\@@href}%
\providecommand \@@href[1]{\endgroup#1\@@endlink}%
\providecommand \@sanitize@url [0]{\catcode `\\12\catcode `\$12\catcode
  `\&12\catcode `\#12\catcode `\^12\catcode `\_12\catcode `\%12\relax}%
\providecommand \@@startlink[1]{}%
\providecommand \@@endlink[0]{}%
\providecommand \url  [0]{\begingroup\@sanitize@url \@url }%
\providecommand \@url [1]{\endgroup\@href {#1}{\urlprefix }}%
\providecommand \urlprefix  [0]{URL }%
\providecommand \Eprint [0]{\href }%
\providecommand \doibase [0]{https://doi.org/}%
\providecommand \selectlanguage [0]{\@gobble}%
\providecommand \bibinfo  [0]{\@secondoftwo}%
\providecommand \bibfield  [0]{\@secondoftwo}%
\providecommand \translation [1]{[#1]}%
\providecommand \BibitemOpen [0]{}%
\providecommand \bibitemStop [0]{}%
\providecommand \bibitemNoStop [0]{.\EOS\space}%
\providecommand \EOS [0]{\spacefactor3000\relax}%
\providecommand \BibitemShut  [1]{\csname bibitem#1\endcsname}%
\let\auto@bib@innerbib\@empty
\bibitem [{\citenamefont {Riemann}(1868)}]{Riemann1}%
  \BibitemOpen
  \bibfield  {author} {\bibinfo {author} {\bibfnamefont {B.}~\bibnamefont
  {Riemann}},\ }\bibfield  {title} {\bibinfo {title} {{{\"U}ber die Hypothesen,
  welche der Geometrie zu Grunde liegen}},\ }\href
  {http://www.deutschestextarchiv.de/riemann_hypothesen_1867} {\bibfield
  {journal} {\bibinfo  {journal} {Abhandlungen der K{\"o}niglichen Gesellschaft
  der Wissenschaften zu G{\"o}ttingen}\ }\textbf {\bibinfo {volume} {13}},\
  \bibinfo {pages} {133} (\bibinfo {year} {1868})}\BibitemShut {NoStop}%
\bibitem [{\citenamefont {Riemann}(1873)}]{Riemann2}%
  \BibitemOpen
  \bibfield  {author} {\bibinfo {author} {\bibfnamefont {B.}~\bibnamefont
  {Riemann}},\ }\bibfield  {title} {\bibinfo {title} {On the hypotheses which
  lie at the bases of geometry},\ }\href {https://doi.org/10.1038/008014a0}
  {\bibfield  {journal} {\bibinfo  {journal} {Nature}\ }\textbf {\bibinfo
  {volume} {8}},\ \bibinfo {pages} {14} (\bibinfo {year} {1873})}\BibitemShut
  {NoStop}%
\bibitem [{\citenamefont {Finsler}(1918)}]{Finsler}%
  \BibitemOpen
  \bibfield  {author} {\bibinfo {author} {\bibfnamefont {P.}~\bibnamefont
  {Finsler}},\ }\emph {\bibinfo {title} {\"{U}ber Kurven und Fl\"{a}chen in
  allgemeinen R\"{a}umen}},\ \href@noop {} {Ph.D. thesis},\ \bibinfo  {school}
  {Georg-August Universit\"{a}t zu G\"{o}ttingen} (\bibinfo {year}
  {1918})\BibitemShut {NoStop}%
\bibitem [{\citenamefont {Tavakol}\ and\ \citenamefont {Van~den
  Bergh}(1986)}]{Tavakol_1986}%
  \BibitemOpen
  \bibfield  {author} {\bibinfo {author} {\bibfnamefont {R.~K.}\ \bibnamefont
  {Tavakol}}\ and\ \bibinfo {author} {\bibfnamefont {N.}~\bibnamefont {Van~den
  Bergh}},\ }\bibfield  {title} {\bibinfo {title} {Viability criteria for the
  theories of gravity and finsler spaces},\ }\href
  {https://doi.org/10.1007/BF00770205} {\bibfield  {journal} {\bibinfo
  {journal} {General Relativity and Gravitation}\ }\textbf {\bibinfo {volume}
  {18}},\ \bibinfo {pages} {849} (\bibinfo {year} {1986})}\BibitemShut
  {NoStop}%
\bibitem [{\citenamefont {Pfeifer}(2019)}]{Pfeifer_2019}%
  \BibitemOpen
  \bibfield  {author} {\bibinfo {author} {\bibfnamefont {C.}~\bibnamefont
  {Pfeifer}},\ }\bibfield  {title} {\bibinfo {title} {{Finsler spacetime
  geometry in Physics}},\ }\href {https://doi.org/10.1142/S0219887819410044}
  {\bibfield  {journal} {\bibinfo  {journal} {Int. J. Geom. Meth. Mod. Phys.}\
  }\textbf {\bibinfo {volume} {16}},\ \bibinfo {pages} {1941004} (\bibinfo
  {year} {2019})},\ \Eprint {https://arxiv.org/abs/1903.10185}
  {arXiv:1903.10185 [gr-qc]} \BibitemShut {NoStop}%
\bibitem [{\citenamefont {Girelli}\ \emph {et~al.}(2007)\citenamefont
  {Girelli}, \citenamefont {Liberati},\ and\ \citenamefont
  {Sindoni}}]{Girelli:2006fw}%
  \BibitemOpen
  \bibfield  {author} {\bibinfo {author} {\bibfnamefont {F.}~\bibnamefont
  {Girelli}}, \bibinfo {author} {\bibfnamefont {S.}~\bibnamefont {Liberati}},\
  and\ \bibinfo {author} {\bibfnamefont {L.}~\bibnamefont {Sindoni}},\
  }\bibfield  {title} {\bibinfo {title} {{Planck-scale modified dispersion
  relations and Finsler geometry}},\ }\href
  {https://doi.org/10.1103/PhysRevD.75.064015} {\bibfield  {journal} {\bibinfo
  {journal} {Phys. Rev.}\ }\textbf {\bibinfo {volume} {D75}},\ \bibinfo {pages}
  {064015} (\bibinfo {year} {2007})},\ \Eprint {https://arxiv.org/abs/0611024}
  {arXiv:0611024 [gr-qc]} \BibitemShut {NoStop}%
\bibitem [{\citenamefont {Amelino-Camelia}\ \emph {et~al.}(2014)\citenamefont
  {Amelino-Camelia}, \citenamefont {Barcaroli}, \citenamefont {Gubitosi},
  \citenamefont {Liberati},\ and\ \citenamefont
  {Loret}}]{Amelino-Camelia:2014rga}%
  \BibitemOpen
  \bibfield  {author} {\bibinfo {author} {\bibfnamefont {G.}~\bibnamefont
  {Amelino-Camelia}}, \bibinfo {author} {\bibfnamefont {L.}~\bibnamefont
  {Barcaroli}}, \bibinfo {author} {\bibfnamefont {G.}~\bibnamefont {Gubitosi}},
  \bibinfo {author} {\bibfnamefont {S.}~\bibnamefont {Liberati}},\ and\
  \bibinfo {author} {\bibfnamefont {N.}~\bibnamefont {Loret}},\ }\bibfield
  {title} {\bibinfo {title} {{Realization of doubly special relativistic
  symmetries in Finsler geometries}},\ }\href
  {https://doi.org/10.1103/PhysRevD.90.125030} {\bibfield  {journal} {\bibinfo
  {journal} {Phys. Rev.}\ }\textbf {\bibinfo {volume} {D90}},\ \bibinfo {pages}
  {125030} (\bibinfo {year} {2014})},\ \Eprint
  {https://arxiv.org/abs/1407.8143} {arXiv:1407.8143 [gr-qc]} \BibitemShut
  {NoStop}%
\bibitem [{\citenamefont {Letizia}\ and\ \citenamefont
  {Liberati}(2017)}]{Letizia:2016lew}%
  \BibitemOpen
  \bibfield  {author} {\bibinfo {author} {\bibfnamefont {M.}~\bibnamefont
  {Letizia}}\ and\ \bibinfo {author} {\bibfnamefont {S.}~\bibnamefont
  {Liberati}},\ }\bibfield  {title} {\bibinfo {title} {{Deformed relativity
  symmetries and the local structure of spacetime}},\ }\href
  {https://doi.org/10.1103/PhysRevD.95.046007} {\bibfield  {journal} {\bibinfo
  {journal} {Phys. Rev.}\ }\textbf {\bibinfo {volume} {D95}},\ \bibinfo {pages}
  {046007} (\bibinfo {year} {2017})},\ \Eprint
  {https://arxiv.org/abs/1612.03065} {arXiv:1612.03065 [gr-qc]} \BibitemShut
  {NoStop}%
\bibitem [{\citenamefont {Ehlers}(2011)}]{Ehlers2011}%
  \BibitemOpen
  \bibfield  {author} {\bibinfo {author} {\bibfnamefont {J.}~\bibnamefont
  {Ehlers}},\ }\bibinfo {title} {General-relativistc kinetic theory of gases},\
  in\ \href {https://doi.org/10.1007/978-3-642-11099-3_4} {\emph {\bibinfo
  {booktitle} {Relativistic Fluid Dynamics}}}\ (\bibinfo  {publisher} {Springer
  Berlin Heidelberg},\ \bibinfo {address} {Berlin, Heidelberg},\ \bibinfo
  {year} {2011})\ pp.\ \bibinfo {pages} {301--388}\BibitemShut {NoStop}%
\bibitem [{\citenamefont {Andreasson}(2011)}]{Andreasson:2011ng}%
  \BibitemOpen
  \bibfield  {author} {\bibinfo {author} {\bibfnamefont {H.}~\bibnamefont
  {Andreasson}},\ }\bibfield  {title} {\bibinfo {title} {{The Einstein-Vlasov
  System/Kinetic Theory}},\ }\href {https://doi.org/10.12942/lrr-2011-4}
  {\bibfield  {journal} {\bibinfo  {journal} {Living Rev. Rel.}\ }\textbf
  {\bibinfo {volume} {14}},\ \bibinfo {pages} {4} (\bibinfo {year} {2011})},\
  \Eprint {https://arxiv.org/abs/1106.1367} {arXiv:1106.1367 [gr-qc]}
  \BibitemShut {NoStop}%
\bibitem [{\citenamefont {Hohmann}\ \emph
  {et~al.}(2020{\natexlab{a}})\citenamefont {Hohmann}, \citenamefont
  {Pfeifer},\ and\ \citenamefont {Voicu}}]{Hohmann:2019sni}%
  \BibitemOpen
  \bibfield  {author} {\bibinfo {author} {\bibfnamefont {M.}~\bibnamefont
  {Hohmann}}, \bibinfo {author} {\bibfnamefont {C.}~\bibnamefont {Pfeifer}},\
  and\ \bibinfo {author} {\bibfnamefont {N.}~\bibnamefont {Voicu}},\ }\bibfield
   {title} {\bibinfo {title} {{Relativistic kinetic gases as direct sources of
  gravity}},\ }\href {https://doi.org/10.1103/PhysRevD.101.024062} {\bibfield
  {journal} {\bibinfo  {journal} {Phys. Rev. D}\ }\textbf {\bibinfo {volume}
  {101}},\ \bibinfo {pages} {024062} (\bibinfo {year} {2020}{\natexlab{a}})},\
  \Eprint {https://arxiv.org/abs/1910.14044} {arXiv:1910.14044 [gr-qc]}
  \BibitemShut {NoStop}%
\bibitem [{\citenamefont {Hohmann}\ \emph
  {et~al.}(2020{\natexlab{b}})\citenamefont {Hohmann}, \citenamefont
  {Pfeifer},\ and\ \citenamefont {Voicu}}]{Hohmann:2020yia}%
  \BibitemOpen
  \bibfield  {author} {\bibinfo {author} {\bibfnamefont {M.}~\bibnamefont
  {Hohmann}}, \bibinfo {author} {\bibfnamefont {C.}~\bibnamefont {Pfeifer}},\
  and\ \bibinfo {author} {\bibfnamefont {N.}~\bibnamefont {Voicu}},\ }\bibfield
   {title} {\bibinfo {title} {{The kinetic gas universe}},\ }\href
  {https://doi.org/10.1140/epjc/s10052-020-8391-y} {\bibfield  {journal}
  {\bibinfo  {journal} {Eur. Phys. J. C}\ }\textbf {\bibinfo {volume} {80}},\
  \bibinfo {pages} {809} (\bibinfo {year} {2020}{\natexlab{b}})},\ \Eprint
  {https://arxiv.org/abs/2005.13561} {arXiv:2005.13561 [gr-qc]} \BibitemShut
  {NoStop}%
\bibitem [{\citenamefont {Pfeifer}\ and\ \citenamefont
  {Wohlfarth}(2012)}]{Pfeifer:2011xi}%
  \BibitemOpen
  \bibfield  {author} {\bibinfo {author} {\bibfnamefont {C.}~\bibnamefont
  {Pfeifer}}\ and\ \bibinfo {author} {\bibfnamefont {M.~N.~R.}\ \bibnamefont
  {Wohlfarth}},\ }\bibfield  {title} {\bibinfo {title} {{Finsler geometric
  extension of Einstein gravity}},\ }\href@noop {} {\bibfield  {journal}
  {\bibinfo  {journal} {Phys.Rev.}\ }\textbf {\bibinfo {volume} {D85}},\
  \bibinfo {pages} {064009} (\bibinfo {year} {2012})},\ \Eprint
  {https://arxiv.org/abs/1112.5641} {arXiv:1112.5641 [gr-qc]} \BibitemShut
  {NoStop}%
\bibitem [{\citenamefont {Hohmann}\ \emph {et~al.}(2019)\citenamefont
  {Hohmann}, \citenamefont {Pfeifer},\ and\ \citenamefont
  {Voicu}}]{Hohmann_2019}%
  \BibitemOpen
  \bibfield  {author} {\bibinfo {author} {\bibfnamefont {M.}~\bibnamefont
  {Hohmann}}, \bibinfo {author} {\bibfnamefont {C.}~\bibnamefont {Pfeifer}},\
  and\ \bibinfo {author} {\bibfnamefont {N.}~\bibnamefont {Voicu}},\ }\bibfield
   {title} {\bibinfo {title} {{Finsler gravity action from variational
  completion}},\ }\href {https://doi.org/10.1103/PhysRevD.100.064035}
  {\bibfield  {journal} {\bibinfo  {journal} {Phys. Rev. D}\ }\textbf {\bibinfo
  {volume} {100}},\ \bibinfo {pages} {064035} (\bibinfo {year} {2019})},\
  \Eprint {https://arxiv.org/abs/1812.11161} {arXiv:1812.11161 [gr-qc]}
  \BibitemShut {NoStop}%
\bibitem [{\citenamefont {Fuster}\ and\ \citenamefont
  {Pabst}(2016)}]{Fuster:2015tua}%
  \BibitemOpen
  \bibfield  {author} {\bibinfo {author} {\bibfnamefont {A.}~\bibnamefont
  {Fuster}}\ and\ \bibinfo {author} {\bibfnamefont {C.}~\bibnamefont {Pabst}},\
  }\bibfield  {title} {\bibinfo {title} {Finsler $pp$-waves},\ }\href
  {https://doi.org/10.1103/PhysRevD.94.104072} {\bibfield  {journal} {\bibinfo
  {journal} {Phys. Rev.}\ }\textbf {\bibinfo {volume} {D94}},\ \bibinfo {pages}
  {104072} (\bibinfo {year} {2016})},\ \Eprint
  {https://arxiv.org/abs/1510.03058} {arXiv:1510.03058 [gr-qc]} \BibitemShut
  {NoStop}%
\bibitem [{\citenamefont {Fuster}\ \emph {et~al.}(2018)\citenamefont {Fuster},
  \citenamefont {Pabst},\ and\ \citenamefont {Pfeifer}}]{Fuster:2018djw}%
  \BibitemOpen
  \bibfield  {author} {\bibinfo {author} {\bibfnamefont {A.}~\bibnamefont
  {Fuster}}, \bibinfo {author} {\bibfnamefont {C.}~\bibnamefont {Pabst}},\ and\
  \bibinfo {author} {\bibfnamefont {C.}~\bibnamefont {Pfeifer}},\ }\bibfield
  {title} {\bibinfo {title} {{Berwald spacetimes and very special
  relativity}},\ }\href {https://doi.org/10.1103/PhysRevD.98.084062} {\bibfield
   {journal} {\bibinfo  {journal} {Phys. Rev.}\ }\textbf {\bibinfo {volume}
  {D98}},\ \bibinfo {pages} {084062} (\bibinfo {year} {2018})},\ \Eprint
  {https://arxiv.org/abs/1804.09727} {arXiv:1804.09727 [gr-qc]} \BibitemShut
  {NoStop}%
\bibitem [{\citenamefont {Caponio}\ and\ \citenamefont
  {Masiello}(2020)}]{Caponio_2020}%
  \BibitemOpen
  \bibfield  {author} {\bibinfo {author} {\bibfnamefont {E.}~\bibnamefont
  {Caponio}}\ and\ \bibinfo {author} {\bibfnamefont {A.}~\bibnamefont
  {Masiello}},\ }\bibfield  {title} {\bibinfo {title} {{On the analyticity of
  static solutions of a field equation in Finsler gravity}},\ }\href
  {https://doi.org/10.3390/universe6040059} {\bibfield  {journal} {\bibinfo
  {journal} {Universe}\ }\textbf {\bibinfo {volume} {6}},\ \bibinfo {pages}
  {59} (\bibinfo {year} {2020})},\ \Eprint {https://arxiv.org/abs/2004.10613}
  {arXiv:2004.10613 [math.DG]} \BibitemShut {NoStop}%
\bibitem [{\citenamefont {Brinkmann}(1925)}]{Brinkmann1925}%
  \BibitemOpen
  \bibfield  {author} {\bibinfo {author} {\bibfnamefont {H.}~\bibnamefont
  {Brinkmann}},\ }\bibfield  {title} {\bibinfo {title} {Einstein spaces which
  are mapped conformally on each other},\ }\href {http://eudml.org/doc/159101}
  {\bibfield  {journal} {\bibinfo  {journal} {Mathematische Annalen}\ }\textbf
  {\bibinfo {volume} {94}},\ \bibinfo {pages} {119} (\bibinfo {year}
  {1925})}\BibitemShut {NoStop}%
\bibitem [{\citenamefont {Bao}\ \emph {et~al.}(2000)\citenamefont {Bao},
  \citenamefont {Chern},\ and\ \citenamefont {Shen}}]{Bao}%
  \BibitemOpen
  \bibfield  {author} {\bibinfo {author} {\bibfnamefont {D.}~\bibnamefont
  {Bao}}, \bibinfo {author} {\bibfnamefont {S.-S.}\ \bibnamefont {Chern}},\
  and\ \bibinfo {author} {\bibfnamefont {Z.}~\bibnamefont {Shen}},\ }\href@noop
  {} {\emph {\bibinfo {title} {An introduction to Finsler-Riemann geometry}}}\
  (\bibinfo  {publisher} {Springer, New York},\ \bibinfo {year}
  {2000})\BibitemShut {NoStop}%
\bibitem [{\citenamefont {Szilasi}(2014)}]{Szilasi}%
  \BibitemOpen
  \bibfield  {author} {\bibinfo {author} {\bibfnamefont {J.}~\bibnamefont
  {Szilasi}},\ }\href@noop {} {\emph {\bibinfo {title} {Connections, Sprays and
  Finsler Structures}}}\ (\bibinfo  {publisher} {World Scientific},\ \bibinfo
  {year} {2014})\BibitemShut {NoStop}%
\bibitem [{\citenamefont {Szilasi}\ \emph {et~al.}(2011)\citenamefont
  {Szilasi}, \citenamefont {Lovas},\ and\ \citenamefont {Cs.}}]{Szilasi2011}%
  \BibitemOpen
  \bibfield  {author} {\bibinfo {author} {\bibfnamefont {J.}~\bibnamefont
  {Szilasi}}, \bibinfo {author} {\bibfnamefont {R.~L.}\ \bibnamefont {Lovas}},\
  and\ \bibinfo {author} {\bibfnamefont {K.~D.}\ \bibnamefont {Cs.}},\
  }\bibfield  {title} {\bibinfo {title} {Several ways to {B}erwald manifolds -
  and some steps beyond},\ }\href@noop {} {\bibfield  {journal} {\bibinfo
  {journal} {Extracta Math.}\ }\textbf {\bibinfo {volume} {26}},\ \bibinfo
  {pages} {89} (\bibinfo {year} {2011})},\ \Eprint
  {https://arxiv.org/abs/1106.2223} {arXiv:1106.2223 [math.DG]} \BibitemShut
  {NoStop}%
\bibitem [{\citenamefont {Pfeifer}\ \emph {et~al.}(2019)\citenamefont
  {Pfeifer}, \citenamefont {Heefer},\ and\ \citenamefont
  {Fuster}}]{Pfeifer:2019tyy}%
  \BibitemOpen
  \bibfield  {author} {\bibinfo {author} {\bibfnamefont {C.}~\bibnamefont
  {Pfeifer}}, \bibinfo {author} {\bibfnamefont {S.}~\bibnamefont {Heefer}},\
  and\ \bibinfo {author} {\bibfnamefont {A.}~\bibnamefont {Fuster}},\
  }\bibfield  {title} {\bibinfo {title} {{Identifying Berwald Finsler
  Geometries}},\ }\href@noop {} {\  (\bibinfo {year} {2019})},\ \Eprint
  {https://arxiv.org/abs/1909.05284} {arXiv:1909.05284 [math.DG]} \BibitemShut
  {NoStop}%
\bibitem [{\citenamefont {Hohmann}\ \emph
  {et~al.}(2020{\natexlab{c}})\citenamefont {Hohmann}, \citenamefont
  {Pfeifer},\ and\ \citenamefont {Voicu}}]{Hohmann:2020mgs}%
  \BibitemOpen
  \bibfield  {author} {\bibinfo {author} {\bibfnamefont {M.}~\bibnamefont
  {Hohmann}}, \bibinfo {author} {\bibfnamefont {C.}~\bibnamefont {Pfeifer}},\
  and\ \bibinfo {author} {\bibfnamefont {N.}~\bibnamefont {Voicu}},\ }\bibfield
   {title} {\bibinfo {title} {{Cosmological Finsler Spacetimes}},\ }\href
  {https://doi.org/10.3390/universe6050065} {\bibfield  {journal} {\bibinfo
  {journal} {Universe}\ }\textbf {\bibinfo {volume} {6}},\ \bibinfo {pages}
  {65} (\bibinfo {year} {2020}{\natexlab{c}})},\ \Eprint
  {https://arxiv.org/abs/2003.02299} {arXiv:2003.02299 [gr-qc]} \BibitemShut
  {NoStop}%
\bibitem [{\citenamefont {Beem}(1970)}]{Beem}%
  \BibitemOpen
  \bibfield  {author} {\bibinfo {author} {\bibfnamefont {J.~K.}\ \bibnamefont
  {Beem}},\ }\bibfield  {title} {\bibinfo {title} {Indefinite {F}insler spaces
  and timelike spaces},\ }\href@noop {} {\bibfield  {journal} {\bibinfo
  {journal} {Can. J. Math.}\ }\textbf {\bibinfo {volume} {22}},\ \bibinfo
  {pages} {1035} (\bibinfo {year} {1970})}\BibitemShut {NoStop}%
\bibitem [{\citenamefont {Asanov}(1985)}]{Asanov}%
  \BibitemOpen
  \bibfield  {author} {\bibinfo {author} {\bibfnamefont {G.~S.}\ \bibnamefont
  {Asanov}},\ }\href@noop {} {\emph {\bibinfo {title} {Finsler Geometry,
  Relativity and Gauge Theories}}}\ (\bibinfo  {publisher} {D. Reidel
  Publishing Company},\ \bibinfo {year} {1985})\BibitemShut {NoStop}%
\bibitem [{\citenamefont {Pfeifer}\ and\ \citenamefont
  {Wohlfarth}(2011)}]{Pfeifer:2011tk}%
  \BibitemOpen
  \bibfield  {author} {\bibinfo {author} {\bibfnamefont {C.}~\bibnamefont
  {Pfeifer}}\ and\ \bibinfo {author} {\bibfnamefont {M.~N.~R.}\ \bibnamefont
  {Wohlfarth}},\ }\bibfield  {title} {\bibinfo {title} {{Causal structure and
  electrodynamics on Finsler spacetimes}},\ }\href@noop {} {\bibfield
  {journal} {\bibinfo  {journal} {Phys.Rev.}\ }\textbf {\bibinfo {volume}
  {D84}},\ \bibinfo {pages} {044039} (\bibinfo {year} {2011})},\ \Eprint
  {https://arxiv.org/abs/1104.1079} {arXiv:1104.1079 [gr-qc]} \BibitemShut
  {NoStop}%
\bibitem [{\citenamefont {Javaloyes}\ and\ \citenamefont
  {Sanchez}(2014)}]{Javaloyes2014-1}%
  \BibitemOpen
  \bibfield  {author} {\bibinfo {author} {\bibfnamefont {M.}~\bibnamefont
  {Javaloyes}}\ and\ \bibinfo {author} {\bibfnamefont {M.}~\bibnamefont
  {Sanchez}},\ }\bibfield  {title} {\bibinfo {title} {Finsler metrics and
  relativistic spacetimes},\ }\href {https://doi.org/10.1142/S0219887814600329}
  {\bibfield  {journal} {\bibinfo  {journal} {Int. J. Geom. Methods Mod.
  Phys.}\ }\textbf {\bibinfo {volume} {11}},\ \bibinfo {pages} {1460032, 15}
  (\bibinfo {year} {2014})}\BibitemShut {NoStop}%
\bibitem [{\citenamefont {Javaloyes}\ and\ \citenamefont
  {S{\'a}nchez}(2014)}]{Javaloyes2014-2}%
  \BibitemOpen
  \bibfield  {author} {\bibinfo {author} {\bibfnamefont {M.}~\bibnamefont
  {Javaloyes}}\ and\ \bibinfo {author} {\bibfnamefont {M.}~\bibnamefont
  {S{\'a}nchez}},\ }\bibfield  {title} {\bibinfo {title} {On the definition and
  examples of {F}insler metrics},\ }\href
  {https://doi.org/10.2422/2036-2145.201203_002} {\bibfield  {journal}
  {\bibinfo  {journal} {Annali della Scuola normale superiore di Pisa, Classe
  di scienze}\ }\textbf {\bibinfo {volume} {13}},\ \bibinfo {pages} {813}
  (\bibinfo {year} {2014})}\BibitemShut {NoStop}%
\bibitem [{\citenamefont {Szab\'o}(1981)}]{Szabo}%
  \BibitemOpen
  \bibfield  {author} {\bibinfo {author} {\bibfnamefont {Z.}~\bibnamefont
  {Szab\'o}},\ }\bibfield  {title} {\bibinfo {title} {Positive definite
  {B}erwald spaces},\ }\href@noop {} {\bibfield  {journal} {\bibinfo  {journal}
  {Tensor}\ } (\bibinfo {year} {1981})}\BibitemShut {NoStop}%
\bibitem [{\citenamefont {Fuster}\ \emph {et~al.}(2020)\citenamefont {Fuster},
  \citenamefont {Heefer}, \citenamefont {Pfeifer},\ and\ \citenamefont
  {Voicu}}]{Fuster_2020}%
  \BibitemOpen
  \bibfield  {author} {\bibinfo {author} {\bibfnamefont {A.}~\bibnamefont
  {Fuster}}, \bibinfo {author} {\bibfnamefont {S.}~\bibnamefont {Heefer}},
  \bibinfo {author} {\bibfnamefont {C.}~\bibnamefont {Pfeifer}},\ and\ \bibinfo
  {author} {\bibfnamefont {N.}~\bibnamefont {Voicu}},\ }\bibfield  {title}
  {\bibinfo {title} {{On the non metrizability of Berwald Finsler
  spacetimes}},\ }\href {https://doi.org/10.3390/universe6050064} {\bibfield
  {journal} {\bibinfo  {journal} {Universe}\ }\textbf {\bibinfo {volume} {6}},\
  \bibinfo {pages} {64} (\bibinfo {year} {2020})},\ \Eprint
  {https://arxiv.org/abs/2003.02300} {arXiv:2003.02300 [math.DG]} \BibitemShut
  {NoStop}%
\bibitem [{\citenamefont {Randers}(1941)}]{Randers}%
  \BibitemOpen
  \bibfield  {author} {\bibinfo {author} {\bibfnamefont {G.}~\bibnamefont
  {Randers}},\ }\bibfield  {title} {\bibinfo {title} {On an asymmetrical metric
  in the four-space of general relativity},\ }\href
  {https://doi.org/10.1103/PhysRev.59.195} {\bibfield  {journal} {\bibinfo
  {journal} {Phys. Rev.}\ }\textbf {\bibinfo {volume} {59}},\ \bibinfo {pages}
  {195} (\bibinfo {year} {1941})}\BibitemShut {NoStop}%
\bibitem [{\citenamefont {Chang}\ and\ \citenamefont
  {Li}(2009)}]{Chang:2009pa}%
  \BibitemOpen
  \bibfield  {author} {\bibinfo {author} {\bibfnamefont {Z.}~\bibnamefont
  {Chang}}\ and\ \bibinfo {author} {\bibfnamefont {X.}~\bibnamefont {Li}},\
  }\bibfield  {title} {\bibinfo {title} {{Modified Friedmann model in
  Randers-Finsler space of approximate Berwald type as a possible alternative
  to dark energy hypothesis}},\ }\href
  {https://doi.org/10.1016/j.physletb.2009.05.001} {\bibfield  {journal}
  {\bibinfo  {journal} {Phys. Lett.}\ }\textbf {\bibinfo {volume} {B676}},\
  \bibinfo {pages} {173} (\bibinfo {year} {2009})},\ \Eprint
  {https://arxiv.org/abs/0901.1023} {arXiv:0901.1023 [gr-qc]} \BibitemShut
  {NoStop}%
\bibitem [{\citenamefont {Kouretsis}\ \emph {et~al.}(2009)\citenamefont
  {Kouretsis}, \citenamefont {Stathakopoulos},\ and\ \citenamefont
  {Stavrinos}}]{Kouretsis:2008ha}%
  \BibitemOpen
  \bibfield  {author} {\bibinfo {author} {\bibfnamefont {A.~P.}\ \bibnamefont
  {Kouretsis}}, \bibinfo {author} {\bibfnamefont {M.}~\bibnamefont
  {Stathakopoulos}},\ and\ \bibinfo {author} {\bibfnamefont {P.~C.}\
  \bibnamefont {Stavrinos}},\ }\bibfield  {title} {\bibinfo {title} {{The
  general very special relativity in Finsler cosmology}},\ }\href
  {https://doi.org/10.1103/PhysRevD.79.104011} {\bibfield  {journal} {\bibinfo
  {journal} {Phys. Rev.}\ }\textbf {\bibinfo {volume} {D79}},\ \bibinfo {pages}
  {104011} (\bibinfo {year} {2009})},\ \Eprint
  {https://arxiv.org/abs/0810.3267} {arXiv:0810.3267 [gr-qc]} \BibitemShut
  {NoStop}%
\bibitem [{\citenamefont {Stavrinos}\ \emph {et~al.}(2014)\citenamefont
  {Stavrinos}, \citenamefont {Vacaru},\ and\ \citenamefont
  {Vacaru}}]{Stavrinos2014}%
  \BibitemOpen
  \bibfield  {author} {\bibinfo {author} {\bibfnamefont {P.}~\bibnamefont
  {Stavrinos}}, \bibinfo {author} {\bibfnamefont {O.}~\bibnamefont {Vacaru}},\
  and\ \bibinfo {author} {\bibfnamefont {S.~I.}\ \bibnamefont {Vacaru}},\
  }\bibfield  {title} {\bibinfo {title} {Modified {E}instein and {F}insler like
  theories on tangent {L}orentz bundles},\ }\href
  {https://doi.org/10.1142/S0218271814500941} {\bibfield  {journal} {\bibinfo
  {journal} {International Journal of Modern Physics D}\ }\textbf {\bibinfo
  {volume} {23}},\ \bibinfo {pages} {1450094} (\bibinfo {year} {2014})},\
  \Eprint
  {https://arxiv.org/abs/http://www.worldscientific.com/doi/pdf/10.1142/S0218271814500941}
  {http://www.worldscientific.com/doi/pdf/10.1142/S0218271814500941}
  \BibitemShut {NoStop}%
\bibitem [{\citenamefont {Voicu}(2010)}]{Voicu:2009wi}%
  \BibitemOpen
  \bibfield  {author} {\bibinfo {author} {\bibfnamefont {N.}~\bibnamefont
  {Voicu}},\ }\bibfield  {title} {\bibinfo {title} {{New considerations on
  Hilbert action and Einstein equations in anisotropic spaces}},\ }\href
  {https://doi.org/10.1063/1.3506066} {\bibfield  {journal} {\bibinfo
  {journal} {AIP Conf. Proc.}\ }\textbf {\bibinfo {volume} {1283}},\ \bibinfo
  {pages} {249} (\bibinfo {year} {2010})},\ \Eprint
  {https://arxiv.org/abs/0911.5034} {arXiv:0911.5034 [gr-qc]} \BibitemShut
  {NoStop}%
\bibitem [{\citenamefont {Minguzzi}(2014)}]{Minguzzi:2014fxa}%
  \BibitemOpen
  \bibfield  {author} {\bibinfo {author} {\bibfnamefont {E.}~\bibnamefont
  {Minguzzi}},\ }\bibfield  {title} {\bibinfo {title} {{The connections of
  pseudo-{F}insler spaces}},\ }\href
  {https://doi.org/10.1142/S0219887814600251, 10.1142/S0219887815920012}
  {\bibfield  {journal} {\bibinfo  {journal} {Int. J. Geom. Meth. Mod. Phys.}\
  }\textbf {\bibinfo {volume} {11}},\ \bibinfo {pages} {1460025} (\bibinfo
  {year} {2014})},\ \bibinfo {note} {[Erratum: Int. J. Geom. Meth. Mod.
  Phys.12,no.7,1592001(2015)]},\ \Eprint {https://arxiv.org/abs/1405.0645}
  {arXiv:1405.0645 [math-ph]} \BibitemShut {NoStop}%
\bibitem [{\citenamefont {Rutz}(1993)}]{Rutz}%
  \BibitemOpen
  \bibfield  {author} {\bibinfo {author} {\bibfnamefont {S.}~\bibnamefont
  {Rutz}},\ }\bibfield  {title} {\bibinfo {title} {A {F}insler generalisation
  of {E}instein's vacuum field equations},\ }\href@noop {} {\bibfield
  {journal} {\bibinfo  {journal} {General Relativity and Gravitation}\ }\textbf
  {\bibinfo {volume} {25}},\ \bibinfo {pages} {1139} (\bibinfo {year}
  {1993})}\BibitemShut {NoStop}%
\bibitem [{\citenamefont {Chen}\ and\ \citenamefont {Shen}(2008)}]{Chen-Shen}%
  \BibitemOpen
  \bibfield  {author} {\bibinfo {author} {\bibfnamefont {B.}~\bibnamefont
  {Chen}}\ and\ \bibinfo {author} {\bibfnamefont {Y.-B.}\ \bibnamefont
  {Shen}},\ }\bibfield  {title} {\bibinfo {title} {{On a class of critical
  Riemann-Finsler metrics}},\ }\href@noop {} {\bibfield  {journal} {\bibinfo
  {journal} {Publ. Math. Debrecen}\ }\textbf {\bibinfo {volume} {72/3-4}},\
  \bibinfo {pages} {451} (\bibinfo {year} {2008})}\BibitemShut {NoStop}%
\bibitem [{\citenamefont {Akbar-Zadeh}(1988)}]{Akbar-Zadeh1988}%
  \BibitemOpen
  \bibfield  {author} {\bibinfo {author} {\bibfnamefont {H.}~\bibnamefont
  {Akbar-Zadeh}},\ }\bibfield  {title} {\bibinfo {title} {Sure les espaces de
  {F}insler \`{a} courbures sectionnelles constantes},\ }\href@noop {}
  {\bibfield  {journal} {\bibinfo  {journal} {Acad.Roy.Belg.Bull.Cl.Sci.}\
  }\textbf {\bibinfo {volume} {74}},\ \bibinfo {pages} {281} (\bibinfo {year}
  {1988})}\BibitemShut {NoStop}%
\bibitem [{\citenamefont {Batista}(2014)}]{Batista_2014}%
  \BibitemOpen
  \bibfield  {author} {\bibinfo {author} {\bibfnamefont {C.}~\bibnamefont
  {Batista}},\ }\bibfield  {title} {\bibinfo {title} {{Killing-Yano Tensors of
  Order n-1}},\ }\href {https://doi.org/10.1088/0264-9381/31/16/165019}
  {\bibfield  {journal} {\bibinfo  {journal} {Class. Quant. Grav.}\ }\textbf
  {\bibinfo {volume} {31}},\ \bibinfo {pages} {165019} (\bibinfo {year}
  {2014})},\ \Eprint {https://arxiv.org/abs/1405.4148} {arXiv:1405.4148
  [gr-qc]} \BibitemShut {NoStop}%
\bibitem [{\citenamefont {Leistner}\ and\ \citenamefont
  {Schliebner}(2016)}]{Leistner2016}%
  \BibitemOpen
  \bibfield  {author} {\bibinfo {author} {\bibfnamefont {T.}~\bibnamefont
  {Leistner}}\ and\ \bibinfo {author} {\bibfnamefont {D.}~\bibnamefont
  {Schliebner}},\ }\bibfield  {title} {\bibinfo {title} {Completeness of
  compact {L}orentzian manifolds with abelian holonomy},\ }\href
  {https://doi.org/10.1007/s00208-015-1270-4} {\bibfield  {journal} {\bibinfo
  {journal} {Mathematische Annalen}\ }\textbf {\bibinfo {volume} {364}},\
  \bibinfo {pages} {1469} (\bibinfo {year} {2016})}\BibitemShut {NoStop}%
\bibitem [{\citenamefont {Blanco}\ \emph {et~al.}(2013)\citenamefont {Blanco},
  \citenamefont {S\'anchez},\ and\ \citenamefont {Senovilla}}]{blanco2011}%
  \BibitemOpen
  \bibfield  {author} {\bibinfo {author} {\bibfnamefont {O.}~\bibnamefont
  {Blanco}}, \bibinfo {author} {\bibfnamefont {M.}~\bibnamefont {S\'anchez}},\
  and\ \bibinfo {author} {\bibfnamefont {J.}~\bibnamefont {Senovilla}},\
  }\bibfield  {title} {\bibinfo {title} {Structure of second-order symmetric
  {L}orentzian manifolds},\ }\href@noop {} {\bibfield  {journal} {\bibinfo
  {journal} {J. Eur. Math. Soc.}\ }\textbf {\bibinfo {volume} {15}},\ \bibinfo
  {pages} {595} (\bibinfo {year} {2013})},\ \Eprint
  {https://arxiv.org/abs/1101.5503} {arXiv:1101.5503 [math.DG]} \BibitemShut
  {NoStop}%
\bibitem [{\citenamefont {McNutt}\ \emph {et~al.}(2009)\citenamefont {McNutt},
  \citenamefont {Coley},\ and\ \citenamefont {Pelavas}}]{Mcnutt2009}%
  \BibitemOpen
  \bibfield  {author} {\bibinfo {author} {\bibfnamefont {D.}~\bibnamefont
  {McNutt}}, \bibinfo {author} {\bibfnamefont {A.}~\bibnamefont {Coley}},\ and\
  \bibinfo {author} {\bibfnamefont {N.}~\bibnamefont {Pelavas}},\ }\bibfield
  {title} {\bibinfo {title} {Isometries in higher-dimensional {CCNV}
  spacetimes},\ }\href {https://doi.org/10.1142/s021988780900359x} {\bibfield
  {journal} {\bibinfo  {journal} {International Journal of Geometric Methods in
  Modern Physics}\ }\textbf {\bibinfo {volume} {06}},\ \bibinfo {pages} {419}
  (\bibinfo {year} {2009})}\BibitemShut {NoStop}%
\bibitem [{\citenamefont {Cohen}\ and\ \citenamefont
  {Glashow}(2006)}]{Cohen:2006ky}%
  \BibitemOpen
  \bibfield  {author} {\bibinfo {author} {\bibfnamefont {A.~G.}\ \bibnamefont
  {Cohen}}\ and\ \bibinfo {author} {\bibfnamefont {S.~L.}\ \bibnamefont
  {Glashow}},\ }\bibfield  {title} {\bibinfo {title} {Very special
  relativity},\ }\href@noop {} {\bibfield  {journal} {\bibinfo  {journal}
  {Phys.Rev.Lett.}\ }\textbf {\bibinfo {volume} {97}},\ \bibinfo {pages}
  {021601} (\bibinfo {year} {2006})},\ \Eprint
  {https://arxiv.org/abs/hep-ph/0601236} {arXiv:hep-ph/0601236 [hep-ph]}
  \BibitemShut {NoStop}%
\bibitem [{\citenamefont {Gibbons}\ \emph {et~al.}(2007)\citenamefont
  {Gibbons}, \citenamefont {Gomis},\ and\ \citenamefont
  {Pope}}]{Gibbons:2007iu}%
  \BibitemOpen
  \bibfield  {author} {\bibinfo {author} {\bibfnamefont {G.}~\bibnamefont
  {Gibbons}}, \bibinfo {author} {\bibfnamefont {J.}~\bibnamefont {Gomis}},\
  and\ \bibinfo {author} {\bibfnamefont {C.}~\bibnamefont {Pope}},\ }\bibfield
  {title} {\bibinfo {title} {General very special relativity is {F}insler
  geometry},\ }\href@noop {} {\bibfield  {journal} {\bibinfo  {journal}
  {Phys.Rev.}\ }\textbf {\bibinfo {volume} {D76}},\ \bibinfo {pages} {081701}
  (\bibinfo {year} {2007})},\ \Eprint {https://arxiv.org/abs/0707.2174}
  {arXiv:0707.2174 [hep-th]} \BibitemShut {NoStop}%
\bibitem [{\citenamefont {Bogoslovsky}(1973)}]{Bogoslovsky1973}%
  \BibitemOpen
  \bibfield  {author} {\bibinfo {author} {\bibfnamefont {G.}~\bibnamefont
  {Bogoslovsky}},\ }\bibfield  {title} {\bibinfo {title} {On a special
  relativistic theory of anisotropic space-time},\ }\href@noop {} {\bibfield
  {journal} {\bibinfo  {journal} {Dokl.Akad. Nauk SSSR}\ ,\ \bibinfo {pages}
  {1055}} (\bibinfo {year} {1973})}\BibitemShut {NoStop}%
\bibitem [{\citenamefont {Bogoslovsky}(1977)}]{Bogoslovsky}%
  \BibitemOpen
  \bibfield  {author} {\bibinfo {author} {\bibfnamefont {G.}~\bibnamefont
  {Bogoslovsky}},\ }\bibfield  {title} {\bibinfo {title} {A
  special-relativistic theory of the locally anisotropic space-time},\
  }\href@noop {} {\bibfield  {journal} {\bibinfo  {journal} {Il Nuovo Cimento B
  Series 11}\ }\textbf {\bibinfo {volume} {40}},\ \bibinfo {pages} {99}
  (\bibinfo {year} {1977})}\BibitemShut {NoStop}%
\bibitem [{\citenamefont {Podolsky}\ \emph {et~al.}(2014)\citenamefont
  {Podolsky}, \citenamefont {Steinbauer},\ and\ \citenamefont
  {Svarc}}]{gyraton}%
  \BibitemOpen
  \bibfield  {author} {\bibinfo {author} {\bibfnamefont {J.}~\bibnamefont
  {Podolsky}}, \bibinfo {author} {\bibfnamefont {R.}~\bibnamefont
  {Steinbauer}},\ and\ \bibinfo {author} {\bibfnamefont {R.}~\bibnamefont
  {Svarc}},\ }\bibfield  {title} {\bibinfo {title} {{Gyratonic pp-waves and
  their impulsive limit}},\ }\href {https://doi.org/10.1103/PhysRevD.90.044050}
  {\bibfield  {journal} {\bibinfo  {journal} {Phys. Rev. D}\ }\textbf {\bibinfo
  {volume} {90}},\ \bibinfo {pages} {044050} (\bibinfo {year} {2014})},\
  \Eprint {https://arxiv.org/abs/1406.3227} {arXiv:1406.3227 [gr-qc]}
  \BibitemShut {NoStop}%
\bibitem [{\citenamefont {{Maluf}}\ \emph {et~al.}(2018)\citenamefont
  {{Maluf}}, \citenamefont {{da Rocha-Neto}}, \citenamefont {{Ulhoa}},\ and\
  \citenamefont {{Carneiro}}}]{Maluf2018}%
  \BibitemOpen
  \bibfield  {author} {\bibinfo {author} {\bibfnamefont {J.~W.}\ \bibnamefont
  {{Maluf}}}, \bibinfo {author} {\bibfnamefont {J.~F.}\ \bibnamefont {{da
  Rocha-Neto}}}, \bibinfo {author} {\bibfnamefont {S.~C.}\ \bibnamefont
  {{Ulhoa}}},\ and\ \bibinfo {author} {\bibfnamefont {F.~L.}\ \bibnamefont
  {{Carneiro}}},\ }\bibfield  {title} {\bibinfo {title} {{Kinetic energy and
  angular momentum of free particles in the gyratonic pp-waves space-times}},\
  }\href {https://doi.org/10.1088/1361-6382/aabd4e} {\bibfield  {journal}
  {\bibinfo  {journal} {Classical and Quantum Gravity}\ }\textbf {\bibinfo
  {volume} {35}},\ \bibinfo {eid} {115001} (\bibinfo {year} {2018})},\ \Eprint
  {https://arxiv.org/abs/1801.04957} {arXiv:1801.04957 [gr-qc]} \BibitemShut
  {NoStop}%
\bibitem [{\citenamefont {Pravda}\ \emph {et~al.}(2002)\citenamefont {Pravda},
  \citenamefont {Pravdova}, \citenamefont {Coley},\ and\ \citenamefont
  {Milson}}]{Pravda:2002us}%
  \BibitemOpen
  \bibfield  {author} {\bibinfo {author} {\bibfnamefont {V.}~\bibnamefont
  {Pravda}}, \bibinfo {author} {\bibfnamefont {A.}~\bibnamefont {Pravdova}},
  \bibinfo {author} {\bibfnamefont {A.}~\bibnamefont {Coley}},\ and\ \bibinfo
  {author} {\bibfnamefont {R.}~\bibnamefont {Milson}},\ }\bibfield  {title}
  {\bibinfo {title} {{All space-times with vanishing curvature invariants}},\
  }\href {https://doi.org/10.1088/0264-9381/19/23/318} {\bibfield  {journal}
  {\bibinfo  {journal} {Class. Quant. Grav.}\ }\textbf {\bibinfo {volume}
  {19}},\ \bibinfo {pages} {6213} (\bibinfo {year} {2002})},\ \Eprint
  {https://arxiv.org/abs/gr-qc/0209024} {arXiv:gr-qc/0209024 [gr-qc]}
  \BibitemShut {NoStop}%
\bibitem [{\citenamefont {Elbistan}\ \emph {et~al.}(2020)\citenamefont
  {Elbistan}, \citenamefont {Zhang}, \citenamefont {Dimakis}, \citenamefont
  {Gibbons},\ and\ \citenamefont {Horvathy}}]{Elbistan:2020mca}%
  \BibitemOpen
  \bibfield  {author} {\bibinfo {author} {\bibfnamefont {M.}~\bibnamefont
  {Elbistan}}, \bibinfo {author} {\bibfnamefont {P.}~\bibnamefont {Zhang}},
  \bibinfo {author} {\bibfnamefont {N.}~\bibnamefont {Dimakis}}, \bibinfo
  {author} {\bibfnamefont {G.}~\bibnamefont {Gibbons}},\ and\ \bibinfo {author}
  {\bibfnamefont {P.}~\bibnamefont {Horvathy}},\ }\bibfield  {title} {\bibinfo
  {title} {{Geodesic motion in Bogoslovsky-Finsler spacetimes}},\ }\href
  {https://doi.org/10.1103/PhysRevD.102.024014} {\bibfield  {journal} {\bibinfo
   {journal} {Phys. Rev. D}\ }\textbf {\bibinfo {volume} {102}},\ \bibinfo
  {pages} {024014} (\bibinfo {year} {2020})},\ \Eprint
  {https://arxiv.org/abs/2004.02751} {arXiv:2004.02751 [gr-qc]} \BibitemShut
  {NoStop}%
\bibitem [{\citenamefont {G{\"u}rlebeck}\ and\ \citenamefont
  {Pfeifer}(2018)}]{Gurlebeck:2018nme}%
  \BibitemOpen
  \bibfield  {author} {\bibinfo {author} {\bibfnamefont {N.}~\bibnamefont
  {G{\"u}rlebeck}}\ and\ \bibinfo {author} {\bibfnamefont {C.}~\bibnamefont
  {Pfeifer}},\ }\bibfield  {title} {\bibinfo {title} {{Observers' measurements
  in premetric electrodynamics I: Time and radar length}},\ }\href@noop {}
  {\bibfield  {journal} {\bibinfo  {journal} {Accepted for publication in PRD}\
  } (\bibinfo {year} {2018})},\ \Eprint {https://arxiv.org/abs/1801.07724}
  {arXiv:1801.07724 [gr-qc]} \BibitemShut {NoStop}%
\bibitem [{\citenamefont {Raetzel}\ \emph {et~al.}(2011)\citenamefont
  {Raetzel}, \citenamefont {Rivera},\ and\ \citenamefont
  {Schuller}}]{Raetzel:2010je}%
  \BibitemOpen
  \bibfield  {author} {\bibinfo {author} {\bibfnamefont {D.}~\bibnamefont
  {Raetzel}}, \bibinfo {author} {\bibfnamefont {S.}~\bibnamefont {Rivera}},\
  and\ \bibinfo {author} {\bibfnamefont {F.~P.}\ \bibnamefont {Schuller}},\
  }\bibfield  {title} {\bibinfo {title} {{Geometry of physical dispersion
  relations}},\ }\href {https://doi.org/10.1103/PhysRevD.83.044047} {\bibfield
  {journal} {\bibinfo  {journal} {Phys. Rev.}\ }\textbf {\bibinfo {volume}
  {D83}},\ \bibinfo {pages} {044047} (\bibinfo {year} {2011})},\ \Eprint
  {https://arxiv.org/abs/1010.1369} {arXiv:1010.1369 [hep-th]} \BibitemShut
  {NoStop}%
\bibitem [{\citenamefont {Bernal}\ \emph {et~al.}(2020)\citenamefont {Bernal},
  \citenamefont {Javaloyes},\ and\ \citenamefont {S\'anchez}}]{Bernal_2020}%
  \BibitemOpen
  \bibfield  {author} {\bibinfo {author} {\bibfnamefont {A.~N.}\ \bibnamefont
  {Bernal}}, \bibinfo {author} {\bibfnamefont {M.~A.}\ \bibnamefont
  {Javaloyes}},\ and\ \bibinfo {author} {\bibfnamefont {M.}~\bibnamefont
  {S\'anchez}},\ }\bibfield  {title} {\bibinfo {title} {Foundations of finsler
  spacetimes from the observers' viewpoint},\ }\href
  {https://doi.org/10.3390/universe6040055} {\bibfield  {journal} {\bibinfo
  {journal} {Universe}\ }\textbf {\bibinfo {volume} {6}},\ \bibinfo {pages}
  {55} (\bibinfo {year} {2020})}\BibitemShut {NoStop}%
\bibitem [{\citenamefont {Kostelecky}(2011)}]{Kostelecky:2011qz}%
  \BibitemOpen
  \bibfield  {author} {\bibinfo {author} {\bibfnamefont {A.}~\bibnamefont
  {Kostelecky}},\ }\bibfield  {title} {\bibinfo {title} {{Riemann-Finsler
  geometry and Lorentz-violating kinematics}},\ }\href
  {https://doi.org/10.1016/j.physletb.2011.05.041} {\bibfield  {journal}
  {\bibinfo  {journal} {Phys. Lett.}\ }\textbf {\bibinfo {volume} {B701}},\
  \bibinfo {pages} {137} (\bibinfo {year} {2011})},\ \Eprint
  {https://arxiv.org/abs/1104.5488} {arXiv:1104.5488 [hep-th]} \BibitemShut
  {NoStop}%
\bibitem [{\citenamefont {Foster}\ and\ \citenamefont
  {Lehnert}(2015)}]{FOSTER2015164}%
  \BibitemOpen
  \bibfield  {author} {\bibinfo {author} {\bibfnamefont {J.}~\bibnamefont
  {Foster}}\ and\ \bibinfo {author} {\bibfnamefont {R.}~\bibnamefont
  {Lehnert}},\ }\bibfield  {title} {\bibinfo {title} {Classical-physics
  applications for {F}insler b space},\ }\href
  {https://doi.org/https://doi.org/10.1016/j.physletb.2015.04.047} {\bibfield
  {journal} {\bibinfo  {journal} {Physics Letters B}\ }\textbf {\bibinfo
  {volume} {746}},\ \bibinfo {pages} {164 } (\bibinfo {year}
  {2015})}\BibitemShut {NoStop}%
\end{thebibliography}%

\end{document}